\documentclass[11pt]{article}
\usepackage[utf8]{inputenc}
\usepackage{palatino}
\usepackage{amssymb}
\usepackage{amsfonts}
\usepackage{amsmath}
\usepackage{amsthm}
\usepackage{latexsym}
\usepackage{color}
\usepackage{epsfig}
\usepackage{hyperref}
\usepackage{bm}
\usepackage{comment}

\newcommand{\ignore}[1]{}
\newenvironment{proofof}[1]{\par{\noindent \bf Proof of #1:}}{\qed\par}

\makeatletter
\renewcommand{\subsubsection}{\@startsection{subsubsection}{3}{0pt}{-12pt}{-5pt}{\normalsize\bf}}
\makeatother

\newcommand{\red}[1]{{{#1}}}
\newtheorem{claim}{Claim}[section]

\newtheorem{proposition}[claim]{Proposition}

\newtheorem{lemma}[claim]{Lemma}
\newtheorem*{theorem*}{Theorem}
\newtheorem*{lemma*}{Lemma}

\newtheorem{theorem}[claim]{Theorem}
\newtheorem{definition}[claim]{Definition}
\newtheorem{corollary}[claim]{Corollary}

\newtheorem{fact}[claim]{Fact}
\newtheorem*{definition*}{Definition}

\newcommand{\dtv}{d_{\mathrm TV}}

\newcommand{\R}{{\bf R}}

\newcommand{\bP}{{\bf P}}
\newcommand{\bX}{{\bf X}}
\newcommand{\bY}{{\bf Y}}
\newcommand{\bQ}{{\bf Q}}

\newcommand{\bR}{{\bf R}}
\newcommand{\bU}{{\bf U}}
\newcommand{\bV}{{\bf V}}

\newcommand{\E}{{\bf E}}

\newcommand{\poly}{\mathrm{poly}}

\newcommand{\fsm}{f_{\mathsf{sm}}}
\newcommand{\gsm}{g_{\mathsf{sm}}}

\begin{document}
 \title{Non interactive simulation of correlated distributions is decidable}
\author{Anindya De\thanks{Supported by a start-up grant from Northwestern University}\\ 
Northwestern University \\{\tt de.anindya@gmail.com} \and Elchanan Mossel\thanks{Supported by ONR grant N00014-16-1-2227   and 
NSF grant CCF 1320105.} \\ MIT \\ {\tt elmos@mit.edu} \and Joe Neeman \\ UT Austin \\ {\tt neeman@iam.uni-bonn.de}}
 \maketitle

\begin{abstract}
A basic problem in information theory is the following:
Let $\bP = (\bX, \bY)$ be an arbitrary distribution where the marginals $\bX$ and $\bY$ are (potentially) correlated.  Let Alice and Bob be two  players where Alice gets samples $\{x_i\}_{i \ge 1}$ and Bob gets samples $\{y_i\}_{i \ge 1}$ and for all $i$,
$(x_i, y_i) \sim \bP$. What joint distributions $\bQ$ can be simulated by Alice and Bob without any interaction? 

Classical works in information theory by G{\'a}cs-K{\"o}rner and Wyner answer this question when at least one of $\bP$ or $\bQ$ is the distribution on $\{0,1\} \times \{0,1\}$ where each marginal is unbiased and identical.  However, other than this special case, the answer to this question is understood in very few cases. Recently, Ghazi, Kamath and Sudan showed that this problem is decidable for $\bQ$ supported on $\{0,1\} \times \{0,1\}$. We extend their result to $\bQ$ supported on any finite alphabet. 

We rely on recent results in Gaussian geometry (by the authors) as well as a new \emph{smoothing argument} inspired by  the method of \emph{boosting} from learning theory and potential function arguments from complexity theory and additive combinatorics.

%

 \end{abstract}
\newpage

\section{Introduction} 

The starting point of this paper is a rather basic problem in information theory and communication complexity,
known as the problem of \emph{non-interactive simulation of joint distributions}: Consider two non-communicating players Alice and Bob. 
Suppose that we give Alice and Bob the sequences $\{\bX_1\}_{i=1}^\infty$ and $\{\bY_i\}_{i=1}^\infty$ respectively, where
the pairs $(\bX_i, \bY_i)$ are independently drawn from some joint distribution $\bP$. Without communicating with each other,
which joint distributions $\bQ$ can Alice and Bob jointly simulate?

To state the problem more precisely, suppose that $\bP$ is a
distribution on $\mathcal{Z} \times \mathcal{Z}$ and that $\bQ$ is a
distribution on $\mathcal{W} \times \mathcal{W}$.  A \emph{non-interactive
strategy} for Alice and Bob simply denotes a triple $(n, f, g)$ such that
$f, g: \mathcal{Z}^n \to \mathcal{W}$, and for which $(f(\bX^n),
g(\bY^n))$ has distribution $\bQ$ whenever $(\bX_i, \bY_i)$ are drawn
independently from $\bP$ (here, $\bX^n$ denotes $\bX_1, \dots, \bX_n$).  The
main question that we consider here is whether a non-interactive strategy exists for a given input
distribution $\bP$ and a given target distribution $\bQ$. Note that not every
pair of input and target distributions admits a non-interactive strategy. The
most obvious example of this is the case where the two coordinates of $\bP$ are
independent; in this case, one can obviously only simulate distributions $\bQ$
whose coordinates are also independent. 

Witsenhausen~\cite{witsenhausen} introduced the problem of non-interactive
simulation, and he studied the case where $\bQ$ is a Gaussian measure on
$\mathbb{R}^2$. In this case, he showed that $\bQ$ can be approximately simulated
by $\bP$ if and only if the absolute value of the correlation between the
components of $\bQ$ is at most the so-called ``maximal correlation
coefficient'' (which we will define later) of $\bP$.
In this case, Witsenhausen showed that for any
$\delta>0$, Alice and Bob can simulate $\bQ$ up to error $\delta$ with $n =
\poly(|\mathcal{Z}|, \log (1/\delta))$. Further, he gave an explicit algorithm
to compute $f$ and $g$ in time $\mathsf{poly}(n)$.


Various other questions of this flavor have been explored in information
theory.  We discuss two examples here.  Let us use $\mathsf{Eq}$ to denote the
distribution supported on $\{0,1\} \times \{0,1\}$ where (i) both marginals are
unbiased Bernoullis and (ii) both marginals are identical with probability $1$. 
\begin{enumerate}
\item In their seminal paper, G{\'a}cs and K{\"o}rner~\cite{gacs1973common}
  studied non-interactive simulation in the case $\bQ = \mathsf{Eq}$.
  In this case, they
  obtained a simple and complete characterization of all $\bP$ such that it is
  possible to non-interactively simulate $\bQ$ from $\bP$. They also
  studied the \emph{simulation capacity}: roughly, how many samples from $\bP$
  are needed to produce each sample from $\bQ$? They showed that the simulation
  capacity is equal to another quantity, which is now known as the
  \emph{G{\'a}cs-K{\"o}rner common information of $\bP$}.

\item Around the same time, Wyner~\cite{wyner1975} considered the complementary problem where $\bP = \mathsf{Eq}$ and $\bQ$ is arbitrary. In other words, Alice and Bob have access to shared random bits and they want to simulate $\bQ$. In this case it is always possible to approximately simulate $\bQ$; Wyner studied the simulation capacity, and showed that it is equal to what is now
  known as the \emph{Wyner common information of $\bQ$}.
\end{enumerate}

When the target distribution $\bQ$ is not Gaussian or $\mathsf{Eq}$, and the
input distribution $\bP$ is not $\mathsf{Eq}$, the problem becomes much more
complicated (see, for example, \cite{KA15} and the references therein).
Nevertheless, the preceding examples suggest that the answer should depend on
how much common information there is between the coordinates of $\bP$ and
between the coordinates of $\bQ$.

To explore this notion more closely, let $\mathbf{G}_{\rho,2}$ be the centered
Gaussian measure on $\mathbb{R}^2$, where each coordinate has unit variance and the correlation
between the coordinates is $\rho>0$.
Consider the setting where $\bP = \mathbf{G}_{\rho,2}$.  If Alice
and Bob want to produce unbiased bits with maximal correlation, Borell's noise
stability inequality~\cite{Borell:85} can be interpreted as saying that the
best they can do is to output the sign of their first input.  In other words, a
pair of unbiased, positively correlated bits can be simulated from $\bP$ if and
only if their correlation is at most $\frac{2}{\pi} \sin^{-1}(\rho)$.

The problem becomes much more difficult if Alice and Bob want to produce
random variables with three equally likely outcomes each. To begin with,
the analogue of Borell's result is not known: we don't know what Alice and
Bob should to in order to maximize the probability that they agree. This
issue was partially addressed in a recent work of the authors~\cite{DMN16a}:
while \cite{DMN16a} does
not solve the simulation problem, \cite{DMN16a} shows that they can approximately compute a strategy
that maximizes the agreement probability, to an arbitrarily small error. Note that this still does not settle the simulation problem,
since joint distributions with three outcomes each are (unlike the case of
two outcomes) not determined by the marginal probabilities and the
agreement probability.

In this work, we extend to framework of~\cite{DMN16a} to answer the
non-interactive simulation problem. Specifically, we show that if $\bQ$ can be
non-interactively simulated from $\bP$ then one can compute, for every $\delta
> 0$, a $\delta$-approximate simulation protocol. Here is an equivalent
formulation, in which $|\bP|$ denotes the size of some standard encoding of $\bP$:
\begin{theorem}~\label{thm:main} 
Let $(\mathcal{Z} \times \mathcal{Z}, \bP)$ and $([k] \times [k], \bQ)$ be probability spaces,
and let $\bX^n = (\bX_1, \dots, \bX_n)$ and $\bY^n = (\bY_1, \dots, \bY_n)$, where $(\bX_i, \bY_i)$ are independently
drawn from $\bP$. For every $\delta > 0$,
there is an algorithm running in time $O_{|\bP|, \delta}(1)$ which distinguishes between the following two cases: 
\begin{enumerate}
\item There exist $n \in \mathbb{N}$ and $f,g: \mathcal{Z}^n \rightarrow [k]$ such that
$\bQ$ and the distribution of $(f(\bX^n), g(\bY^n))$ are $\delta$-close
in total variation distance.  In this case, there is an explicit $n_0 =
n_0(|\bP|, \delta)$ such that we may choose $n \le n_0$. Further, the functions
$f$ and $g$ can be explicitly computed. 
\item For any $n \in \mathbb{N}$ and $f,g: \mathcal{Z}^n \rightarrow [k]$, $\bQ$
and the distribution of $(f(\bX^n), g(\bY^n))$ are $8\delta$-far in total variation distance.
\end{enumerate}
\end{theorem}

We remark here that the bound $n_0$, while computable, is not primitive
recursive and has an Ackermann type growth, which is introduced by our
application of a regularity lemma from~\cite{DS14}.
It is easy to see that to prove Theorem~\ref{thm:main}, it suffices to prove the following theorem. 
\begin{theorem}~\label{thm:junta} 
  With the notation of Theorem~\ref{thm:main}, suppose there exist
  $f, g: \mathcal{Z}^n \to [k]$ such that $(f(\bX^n), g(\bY^n)) \sim \bQ$.
  Then, there exist $n_0 = n_0 (|\mathbf{P}|, \delta)$ and $f_{\delta},
  g_\delta : \mathcal{Z}^{n_0} \rightarrow [k]$ such that $\bQ$ and the
  distribution of $(f_{\delta}(\bX^{n_0}), g_{\delta}(\bY^{n_0}))$ are
  $\delta$-close in total variation distance. Moreover, $n_0$ is computable. Further, the functions $f_\delta$ and $g_\delta$ can be explicitly computed. 
\end{theorem}
The gist of the above theorem is that if a distribution
can be simulated then it can be approximately simulated with a bounded number
of samples. (The crucial point in the previous sentence is that the bound
is explicit, and that it depends only on $\bP$ and the desired accuracy.)


\subsection{Recent work, and the difficulty of going from two to three} 
In a recent paper, Ghazi, Kamath, and Sudan~\cite{GKS16} proved
Theorems~\ref{thm:main} and~\ref{thm:junta} in the case $k=2$.  Moreover, they
gave an explicit doubly exponential bound on $n_0$ and the running time of the
algorithm.  Borell's noise stability theorem (which is not available for $k >
2$) played an important role in their analysis.  To explain the bottleneck in
extending their result for any $k$, we will elaborate on the case where $\mathcal{Z} =
\mathbb{R}$ and $\bP = \mathbf{G}_{\rho,2}$.
We begin by recalling Borell's inequality~\cite{Borell:85} on Gaussian noise stability.
\begin{theorem}\label{thm:Borell} \cite{Borell:85}
  Let $\bP = \mathbf{G}_{\rho,2}$. 
For any $\mu_1, \mu_2 \in [0,1]$, let $f, g : \mathbb{R}^n \rightarrow \{0,1\}$ such that  $\mathbf{E}[f] = \mu_1$ and $\mathbf{E}[g] = \mu_2$. Let us choose $\kappa_1, \kappa_2$ such that for $f_{\mathsf{LTF}}, g_{\mathsf{LTF}} : \mathbb{R} \rightarrow \{0,1\}$ defined as $f_{\mathsf{LTF}}(x) = \mathsf{sign}(x- \kappa_1)$ and $g_{\mathsf{LTF}}(x) = \mathsf{sign}(x- \kappa_2)$, we have $\mathbf{E}[f_{\mathsf{LTF}}] = \mu_1$ and $\mathbf{E}[g_{\mathsf{LTF}}] = \mu_2$. Then, 
$$\Pr_{(\bX,\bY) \sim \bP} [f_{\mathsf{LTF}}(\bX) = g_{\mathsf{LTF}}(\bY)] \geq \Pr_{(\bX,\bY) \sim \bP} [f(\bX) = g(\bY)].  $$
Likewise, if we define $h_{\mathsf{LTF}} = \mathsf{sign} (-x - \kappa_2)$, then $\mathbf{E}[h_{\mathsf{LTF}}] = \mu_2$ and 
$$\Pr_{(\bX,\bY) \sim \bP} [f_{\mathsf{LTF}}(\bX) = h_{\mathsf{LTF}}(\bY)] \leq \Pr_{(\bX,\bY) \sim \bP} [f(\bX) = g(\bY)].  $$
\end{theorem} 
To explain the intuitive meaning of these theorems,
let us define $\mathsf{Corr}_{\max}(\rho, \mu_1, \mu_2)$ and $\mathsf{Corr}_{\min}(\rho, \mu_1, \mu_2)$ as 
\[
\mathsf{Corr}_{\max}(\rho, \mu_1, \mu_2)= \Pr_{(\bX,\bY) \sim \bP} [f_{\mathsf{LTF}}(\bX) = g_{\mathsf{LTF}}(\bY)], 
\]
\[
\mathsf{Corr}_{\min}(\rho, \mu_1, \mu_2)= \Pr_{(\bX,\bY) \sim \bP} [f_{\mathsf{LTF}}(\bX) = h_{\mathsf{LTF}}(\bY)]
\]
where $f_{\mathsf{LTF}}$, $g_{\mathsf{LTF}}$ and
$h_{\mathsf{LTF}}$ are halfspaces defined in Theorem~\ref{thm:Borell}. 
Then, Borell's result implies that for any given measures $\mu_1, \mu_2$ and functions $f, g$ with these measures,  the probability that  $f(\bX)$ and $g(\bY)$ are identical lies between $\mathsf{Corr}_{\max}(\rho, \mu_1, \mu_2)$ and $\mathsf{Corr}_{\min}(\rho, \mu_1, \mu_2)$. Further, now, it easily follows that for any $\eta$ such that 
$\mathsf{Corr}_{\min }(\rho, \mu_1, \mu_2) \le \eta \le \mathsf{Corr}_{\max}(\rho, \mu_1, \mu_2)$, there is a function $g_{\eta}: \mathbb{R} \rightarrow \{0,1\}$ such that $\mathbf{E}[g_{\eta}] = \mu_2$ and $
\eta = \Pr_{(\bX, \bY) \sim P} [f(\bX) = g_{\eta}(\bY)]$. In fact, it is also easy to see that $g_\eta$ can be assumed to be the indicator function of an interval. 

Now, consider any distribution $\bQ$ on $\{0,1\} \times \{0,1\}$, and take $(\bU, \bV) \sim \bQ$. Assume that 
there exist $f, g: \mathbb{R}^n \rightarrow \{0,1\}$ 
such that $(f(\bX^n) , g(\bY^n)) \sim \bQ$.  Defining $\mu_{1, \bQ} = \mathbf{E}[\bU]$, $\mu_{2, \bQ} = \mathbf{E}[\bV]$ and $\eta_{\bQ} = \Pr[\bU = \bV]$ and applying Theorem~\ref{thm:Borell}, we obtain that there are functions $f_\bQ,g_\bQ : \mathbb{R} \rightarrow \{0,1\}$ which satisfy
\[
  \mathbf{E}[f_\bQ (\bX) ] = \mu_{1, \bQ}, \quad \mathbf{E}[g_\bQ (\bY) ] = \mu_{2, \bQ},
\]
and 
\[
  \Pr_{(\bX, \bY) \sim \bP} [f_\bQ(\bX) = g_{\bQ}(\bY)]  = \eta_\bQ.
\]
Further, the functions $f_{\bQ}$ and $g_{\bQ}$ are in fact indicators of intervals and given  $\mu_{1, \bQ}$, 
$\mu_{2, \bQ}$ and $\eta_{\bQ}$, the functions $f_{\bQ}$ and $g_{\bQ}$ can be explicitly computed. 
Observe that any distribution  $\bQ$ over $\{0,1\} \times \{0,1\}$ is characterized by the quantities $\mu_{1, \bQ}$, $\mu_{2, \bQ}$ and $\eta_{\bQ}$. Thus, it implies that $(f_\bQ(\bX), g_{\bQ}(\bY)) \sim \bQ$. This completely settles the non-interactive simulation problem in the case $k=2$, when $\bP$ is the Gaussian measure $\mathbf{G}_{\rho,2}$ on $\mathbb{R}^2$.

In particular, we see that when $\bP$ is Gaussian,
the result of~\cite{GKS16} is a straightforward consequence of Theorem~\ref{thm:Borell}.
Indeed, their main contribution was to show that the
general case reduces to the Gaussian case. Moreover, that part of their argument
turns out to generalize to $k > 2$ (as we will discuss later).
Therefore, let us continue examining the
case where $\bP$ is Gaussian, and see why $k > 2$ causes trouble. There are
two problems:
\begin{enumerate}
\item The analogue of Borell's result for $k>2$ is not known. In particular, the following simple question is still open:  let $\boldsymbol{\mu} \in \Delta_k$ where $\Delta_k$ is the convex hull of the standard unit vectors $\{ \mathbf{e}_1, \ldots, \mathbf{e}_k\}$. Let 
$A_{\boldsymbol{\mu}} = \{f : \mathbb{R}^n \rightarrow [k] : \mathbf{E}[f] =\boldsymbol{\mu} \}$. Among all $f \in  A_{\boldsymbol{\mu}}$, what $f$ maximizes the probability $\Pr_{(\bX,\bY) \sim \bP} [f(\bX) = f(\bY)]$? If $k=2$, then Theorem~\ref{thm:Borell} asserts that $f$ is the indicator of some halfspace; for $k > 3$, the answer is almost completely unknown.
Of particular relevance to us, it is not even known whether the optimal value can be achieved in any finite dimension
(whereas in the case $k=2$, it is achieved in one dimension).
\item For $k=2$, any distribution $\bR = (\bR_1, \bR_2)$ supported on $[k] \times [k]$ is completely defined by $\mathbf{E}[\bR_1]$, $\mathbf{E}[\bR_2]$ and $\Pr[\bR_1 = \bR_2]$. However, this is no longer true  when $k>2$. 
\end{enumerate}

In \cite{DMN16a}, the authors partially circumvented the first issue. To explain the result of \cite{DMN16a}, we will need to introduce two notions. The first is that of the (standard) Ornstein-Uhlenbeck noise operator. Namely, for any $t \ge 0$ and $f: \mathbb{R}^n \rightarrow \mathbb{R}$, we define $P_t f : \mathbb{R}^n \rightarrow \mathbb{R}$ as 
\begin{equation}\label{eq:noise-operator-def}
P_tf(x) = \mathop{\mathbf{E}}_{y \sim \gamma_n} [f(e^{-t} x +\sqrt{1-e^{-2t}} y)]. 
\end{equation}
To see the connection between $P_t$ and our $\rho$-correlated Gaussian distribution $\mathbf{P}=\mathbf{G}_{\rho,2}$,
choose $t$ so that $e^{-t} = \rho$. Then
\[
\mathbf{E}_{(\bX, \bY)^n \sim \bP^n} [f(\bX^n) \cdot f(\bY^n)] = \mathbf{E}_{\bX^n \sim \gamma_n} [f(\bX^n) \cdot P_t f (\bX^n)]. 
\]
The above quantity is often referred to as the noise stability of $f$ at noise rate $t>0$. 
Note that the operator $P_t$ is a linear operator on the space of functions mapping $\mathbb{R}^n$ to $\mathbb{R}$. In fact, the noise operator can be syntactically extended to functions $f: \mathbb{R}^n \rightarrow \mathbb{R}^k$  with the same definition as in (\ref{eq:noise-operator-def}). Embedding $\Delta_k$ in $\mathbb{R}^k$ and identifying $[k]$ with the vertices of $\Delta_k$, we obtain that 
\[
\mathbf{E}_{(\bX, \bY)^n \sim \bP^n} [\langle f(\bX^n) , f(\bY^n) \rangle] = \mathbf{E}_{\bX^n \sim \gamma_n} [\langle f(\bX^n) , P_t f (\bX^n) \rangle].
\]

Let 
us now recall the notion of a multivariate polynomial threshold function (PTF) from \cite{DMN16a}. Given polynomials, $p_1, \ldots, p_k: \mathbb{R}^n \rightarrow \mathbb{R}$, define $f= \mathsf{PTF}(p_1, \ldots, p_k)$ as
\[
f(x) = \begin{cases} j &\textrm{if }p_j(x)>0\textrm{ and }p_i(x) \le 0 \textrm{ for all }j\not =i, \\
1 &\textrm{ otherwise}. \\
\end{cases}
\]
In \cite{DMN16a}, the authors proved the following theorem. A notation we will adopt for the rest of the paper is that unless explicitly mentioned otherwise, the expectation is always w.~r.~t.~the variable being a standard Gaussian where the ambient dimension will be clear from the context. 
\begin{theorem}~\label{thm:DMN1}
Let $f: \mathbb{R}^n \rightarrow [k]$ such that $\mathbf{E}[f] = \boldsymbol{\mu} \in \mathbb{R}^k$. Then, given any $t>0, \epsilon>0$, there exists an explicitly computable $n_0 = n_0(t, k, \epsilon)$ and $d= d(t, k, \epsilon)$ such that there is a degree-$d$ PTF $g: \mathbb{R}^{n_0} \rightarrow [k]$ with
\begin{enumerate}
\item $\Vert  \mathbf{E}[g] - \boldsymbol{\mu} \Vert_1 \le \epsilon$. 
\item $\mathbf{E}[\langle g, P_t g \rangle] \ge \mathbf{E}[\langle f, P_t f \rangle] - \epsilon$.
\end{enumerate}
\end{theorem}
In other words, Theorem~\ref{thm:DMN1} shows that for any given $\boldsymbol{\mu}$ and error parameter $\epsilon>0$, there is a low-degree, low-dimensional PTF $g$ which approximately maximizes the noise stability and whose expectation is close to $\boldsymbol{\mu}$. {\red{We remark here that the issue of matching the expectation exactly versus approximately is insignificant since expectations can always be made to match exactly by suffering a tiny change in the correlation.}}
The proof of Theorem~\ref{thm:DMN1} has two separate steps: 
\begin{enumerate}
  \item (\textbf{Smooth}) The first step is to show that given any $f: \mathbb{R}^n \rightarrow [k]$ with $\mathbf{E}[f] = \boldsymbol{\mu}$, there is a  degree $d = d(t, k, \epsilon)$ PTF $h$ on $n$ variables such that $\Vert \mathbf{E}[h] -\boldsymbol{\mu}\Vert_1 \le \epsilon$ and $\mathbf{E}[\langle h, P_t h \rangle] \ge \mathbf{E}[\langle f, P_tf \rangle]-\epsilon$. In other words, reduce the degree but not the dimension.
  
  The main idea here is to modify the function $f$ by first smoothing it and
  then rounding it back to the discrete set $[k]$. It is fairly easy to show
  that this procedure doesn't decrease the noise stability of $f$ (as long as
  the amount of smoothing is chosen to match the noise parameter $t$). The more
  difficult part is to show that the result of this procedure is close to a
  low-degree PTF. This is done using a randomized rounding argument: we show that by
  rounding the smoothed function at a random threshold, the expected Gaussian
  surface area of the resulting partition is bounded; in particular, there
  exists a good way to round. A well-known link between Gaussian surface area
  and Hermite expansions then implies that the rounded, smoothed function is
  almost a low-degree PTF. {\red{This argument uses the co-area formula, gradient bounds and is inspired by ideas from \cite{KNOW14, Neeman14}.}} 

\item (\textbf{Reduce}) The second step is to show that given a multivariate PTF $h$, there is a multivariate PTF $g$ on $n_0 = n_0(t,k, \epsilon)$ variables such that the noise stability of $g$ is the same as that of the noise stability of $h$ up to an additive error $\epsilon$. This step uses several ideas and results from \cite{DS14}. To give a brief overview of this part, 
we start with the notion of an \emph{eigenregular} 
polynomial which was introduced in \cite{DS14}. 
A polynomial is said to be $\delta$-eigenregular if for the canonical tensor $\mathcal{A}_p$ associated with the polynomial, the ratio of  the maximum singular value to its Frobenius norm is at most $\delta$.
Let us assume that $h = \mathsf{PTF}(p_1, \ldots, p_k)$. The
\emph{regularity lemma} from \cite{DS14}, roughly speaking, shows that each of the polynomials $p_1, \ldots, p_k$
can be written as a low-degree ``outer'' polynomial composed with a bounded number of $\delta$-eigenregular, low-degree
``inner'' polynomials. Using the central limit theorem from~\cite{DS14} and several other new
technical ingredients, one can replace the whole collection of inner polynomials by a new collection of inner polynomials
on a bounded number of variables. Moreover, one can do this replacement while hardly affecting the distribution
of the outer polynomial. In particular, this whole procedure constructs a new PTF on a bounded number of inputs, and with
approximately the same noise stability as the original PTF.

\end{enumerate}

\textbf{How to prove Theorem~\ref{thm:junta}:}
We will first outline the proof of Theorem~\ref{thm:junta} in the case that $\bP = \mathbf{G}_{\rho,2}$
(the $\rho$-correlated Gaussian measure on $\mathbb{R}^2$).
As we
observed earlier, any function with codomain $[k]$ naturally maps to
$\mathbb{R}^k$ by identifying $i \in [k]$ with the standard unit vector
$\mathbf{e}_i \in \mathbb{R}^k$. Also, for any function $f: \mathbb{R}^n
\rightarrow \mathbb{R}^k$ and $1 \le j \le k$, we let $f_j: \mathbb{R}^n
\rightarrow \mathbb{R}$ denote the $j^{th}$ coordinate of $f$. Then, observe
that for all $1 \le i, j \le k$, 
\[
\Pr_{(\bX^n, \bY^n) \sim \bP^n} [f(\bX^n) = i  \wedge g(\bY^n) = j ] = \mathbf{E}[f_{i} P_t g_{j}]. 
\]
In particular, to prove Theorem~\ref{thm:junta} in the case $\bP = \mathbf{G}_{\rho,2}$ it suffices to prove an improvement of
Theorem~\ref{thm:DMN1}, where the inequality 
$\mathbf{E}[\langle g, P_t g \rangle] \ge \mathbf{E}[\langle f, P_t f \rangle] - \epsilon$ is replaced
by an almost-equality: $|\mathbf{E} [g_i P_t g_j] - \mathbf{E} [f_i P_t f_j]| \le \epsilon$ for all $i, j$.
In fact, we will prove something slightly stronger, by starting with two functions instead of one.

The proof of Theorem~\ref{thm:junta} will follow the same smooth/reduce outline as the proof of Theorem~\ref{thm:DMN1}.
Moreover, the ``reduce'' step will be essentially the same as the one in~\cite{DMN16a}. Therefore, we will outline
only the ``smooth'' step.
Define the set $\Delta_{k,\epsilon}$ as
\[
  \Delta_{k,\epsilon} = \{x \in \mathbb{R}^k : \exists  y \in \Delta_k, \ \ \Vert x - y \Vert_1 \le \epsilon\}.
\]
Thus, if $\epsilon = 0$, then $\Delta_{k,\epsilon} = \Delta_k$.
In the ``smooth'' step for the proof of Theorem~\ref{thm:junta}, we will show that for any pair $f$, $g$ of
functions $\R^n \to [k]$, there exist functions $\tilde f, \tilde g: \R^n \to \R^k$ such that

For every $\epsilon>0$, we will show that there are functions $f_1, g_1: \mathbb{R}^n \rightarrow \mathbb{R}^k$ satisfying the following conditions:
\begin{itemize}
  \item[(i)] $ \Vert \mathbf{E}[f] - \mathbf{E}[\tilde f] \Vert_1 \le \epsilon$, $ \Vert \mathbf{E}[g] - \mathbf{E}[\tilde g] \Vert_1 \le \epsilon$;
  \item[(ii)] the functions $f_1, g_1$ are linear combinations of $O_{k,t,\epsilon}(1)$ low-degree PTFs (with some special structure that we will describe later);
\item[(iii)] $\Pr[\tilde f(\bX^n) \in \Delta_{k,\epsilon}] \ge 1 - \epsilon$ and $\Pr[\tilde g(\bY^n) \in \Delta_{k,\epsilon}] \ge 1-\epsilon$; and
\item[ (iv)] for any $1 \le i, j \le k$, $\big| \mathbf{E}[\langle f_{i} P_t g_{j} \rangle] - \mathbf{E}[\langle \tilde f_{i} P_t \tilde g_{j} \rangle ] \big| \le \epsilon$.  
\end{itemize}

The precise statement corresponding to this step is given in
Lemma~\ref{lem:smoothing}, which contains most of the technically new ideas in
the paper.  In particular, we employ a new ``boosting'' based idea to obtain
the functions $\tilde f$ and $\tilde g$. 

The proof of Lemma~\ref{lem:smoothing} comes in two main steps. We start with arbitrary functions $f$ and $g$.
First, we show that there are projections of polynomial threshold functions $\fsm$ and $\gsm$ which have the same low-level Hermite spectrum as $f$ and $g$. This  is carried out in an iterative argument using a potential function, 
and is inspired by similar iterative algorithms appearing in boosting~\cite{Schapire:90, Fre95} from learning theory, the hardcore lemma in complexity theory~\cite{Imp95} and dense model theorems in graph theory~\cite{Frieze1999} and additive combinatorics~\cite{Tao:07, TTV09:conf}.  
While these iterative algorithms have recently been used to prove structural results in complexity theory~\cite{DDFS14, LRS15, TTV09:conf}, since our algorithm is in the multidimensional setting, it is somewhat more delicate than these applications. 
The main argument here is carried out in Lemma~\ref{lem:Boosting}, and we bound the degree of the resulting polynomials in Corollary~\ref{corr:fsm}.

The next step is to show that we can replace the projected polynomial threshold
functions by polynomials that with high probability take values very close to
the simplex (call them $\fsm'$ and $\gsm'$).
This is carried out in Lemma~\ref{lem:smoothing-1}, using
Bernstein approximations for
Lipschitz functions. Finally,
we use some probabilistic tricks to replace $\fsm'$ and $\gsm'$ by functions $\tilde f$ and
$\tilde g$ which are linear combinations of low-degree PTFs. This finishes the
proof of Lemma~\ref{lem:smoothing}.

\subsection{What happens when $\bP$ is not Gaussian?}~\label{sec:non-gauss}
{{So far,  the discussion has pertained to the case when $\bP = \mathbf{G}_{\rho,2}$.  What happens if $\bP$ is a different probability distribution? 

As we have remarked earlier, the main result of~\cite{GKS16} is that the $k=2$
case of Theorem~\ref{thm:junta} essentially reduces to the special case $\bP =
\mathbf{G}_{\rho,2}$.  Their argument uses quite general tools from Boolean
function analysis such as the invariance principle~\cite{MOO10, Mossel2010} and
regularity lemmas for low-degree polynomials~\cite{DSTW:10, DDS14}.
A similar argument can be used to prove Theorem~\ref{thm:junta}
by reducing to the Gaussian case; however, we will actually need a slightly
stronger Gaussian version of Theorem~\ref{thm:junta}:


\begin{theorem}\label{thm:junta-strong}
Let $\bP = \mathbf{G}_{\rho,2}$ and let $f^{(1)}, \ldots, f^{(\ell)}: \mathbb{R}^n \rightarrow [k]$ and $g^{(1)}, \ldots, g^{(\ell)}: \mathbb{R}^n \rightarrow [k]$ where we  define $\bQ_{i,j}$ as 
$\bQ_{i,j} = (f^{(i)}(\bX^n), g^{(j)}(\bY^n))$. 
Then, for every $\delta>0$, there is an explicitly defined constant 
$n_0 = n_0(\ell, k, \delta)$ and explicitly defined functions $f^{(1)}_{\mathsf{junta}}, \ldots, f^{(\ell)}_{\mathsf{junta}}: \mathbb{R}^{n_0} \rightarrow [k]$  and $g^{(1)}_{\mathsf{junta}}, \ldots, g^{(\ell)}_{\mathsf{junta}}: \mathbb{R}^{n_0} \rightarrow [k]$ such that for every $1 \le i, j \le \ell$, 
$\dtv((f^{(i)}_{\mathsf{junta}}(\bX^{n_0}), g^{(j)}_{\mathsf{junta}}(\bY^{n_0})), \bQ_{i,j}) \le \delta$. 
\end{theorem}

Note that the $\ell=1$ case of Theorem~\ref{thm:junta-strong} is exactly the
$\bP = \mathbf{G}_{\rho,2}$ case of Theorem~\ref{thm:junta}, the proof of which
we outlined above. Then $\ell > 1$ case has essentially the same proof, but
with more notation.


In order to prove Theorem~\ref{thm:junta} from Theorem~\ref{thm:junta-strong},
Alice and Bob both execute a ``decision tree.'' By standard arguments from
Boolean function analysis (see~\cite{ODonnell:book} for definitions of the
terminology that follows), Alice and Bob can represent $f$ and $g$ by small
decision trees, such that most of the ``leaf'' functions (call them
$\{f^{(i)}\}_{1 \le i \le \ell}$ and $\{g^{(i)}\}_{1 \le i \le \ell}$) are
\emph{low-influence} functions.  The invariance principle of Mossel
\emph{et~al.}~\cite{MOO10, Mossel2010} allows us to replace $\{f^{(i)}\}_{1 \le
i \le \ell}$ and $\{g^{(i)}\}_{1 \le i \le \ell}$ by functions of Gaussian
variables; essentially, we can pretend that Alice and Bob have access to
independent copies of $\mathbf{G}_{\rho,2}$ where $\rho$ is the so-called
maximal correlation coefficient of $(\bX, \bY)$. Finally, we apply
Theorem~\ref{thm:junta-strong} to this collection of Gaussian ``leaf''
functions.  In the end, we have replaced Alice and Bob's initial functions by a
pair of decision trees of bounded size, where every leaf function is a function
of a bounded number of Gaussian variables.
We give a more detailed overview  of this reduction in
Section~\ref{section:GKS}.

\subsection{Acknowledgements}
We thank
Pritish Kamath, Badih Ghazi and Madhu Sudan for pointing out that the $\ell=1$ case
of Theorem~\ref{thm:junta-strong} is not sufficient to derive Theorem~\ref{thm:junta}. 
(An earlier version of this paper incorrectly claimed that it was.)
We also thank the anonymous reviewers
who pointed out the same gap.

}}



\section{Technical preliminaries} 

We will start by defining some technical preliminaries which will be useful for the rest of the paper. 
\begin{definition}
For $k \in \mathbb{N}$ and $1\le i \le k$, let $\mathbf{e}_i$ be the unit vector along coordinate $i$ and let $\Delta_k$ be the convex hull formed by $\{ \mathbf{e}_i \}_{1\le i \le k}$. 
\end{definition}
In this paper, we will be working on the space of functions $f: \mathbb{R}^n \rightarrow \mathbb{R}$ where the domain is equipped with the standard $n$ dimensional normal measure (denoted by $\gamma_n(\cdot)$). Unless explicitly mentioned otherwise, all the functions considered in this paper will be in $L^2(\gamma_n)$. A key property of such functions is that they admit the so-called Hermite expansion. Let us define a family of polynomials $H_q: \mathbb{R} \rightarrow \mathbb{R}$ (for $q \ge 0$) as 
$$
H_0(x) = 1 ; \ H_1(x) = x ; \  H_q(x) = \frac{(-1)^q}{\sqrt{q!}}  \cdot e^{x^2/2}  \cdot \frac{d^q}{dx^q} e^{-x^2/2}. 
$$
Let $\mathbb{Z}^{\ast}$ denote the subset of non-negative integers and $S \in \mathbb{Z}^{\ast n}$. Define $H_S: \mathbb{R}^n \rightarrow \mathbb{R}$ as
$$
H_S(z) = \prod_{i=1}^n H_{S_i}(z_i). 
$$
It is well known that the set $\{H_S\}_{S \in \mathbb{Z}^{\ast n}}$ forms an orthonormal basis for $L^2(\gamma_n)$.  In other words, every $f \in L^2(\gamma_n)$
may be written as
$$
f = \sum_{S \in \mathbb{Z}^{\ast n}} \widehat{f}(S) \cdot H_S,
$$
where $\widehat{f}(S)$ are typically referred to as the \emph{Hermite coefficients} and expansion is referred to as the \emph{Hermite expansion}. The notion of Hermite expansion can be easily extended to $f: R^n \rightarrow \mathbb{R}^k$ as follows: Let $f = (f_1, \ldots, f_k)$ and let 
$$
f_i  =  \sum_{S \in \mathbb{Z}^{\ast n}} \widehat{f_i}(S) \cdot H_S. 
$$
Then, the Hermite expansion of $f$ is given by $\sum_{S \in \mathbb{Z}^{\ast n}} \widehat{f}(S) \cdot H_S$ where $\widehat{f}(S) = (\widehat{f_1}(S), \ldots, \widehat{f_k}(S))$. 
In this setting, we also have Parseval's identity:
\begin{equation}\label{eq:parseval}
\int \Vert f (x) \Vert_2^2  \ \gamma_n(x) dx = \sum_{S \in \mathbb{Z}^{\ast n}} \Vert \widehat{f}(S) \Vert_2^2
\end{equation}
For $f: \mathbb{R}^n \rightarrow \mathbb{R}^k$ and $d \in \mathbb{N}$, define
$f_{\le d} : \mathbb{R}^n \rightarrow \mathbb{R}^k$ by
$$
f_{\le d} (x) = \sum_{S: |S| \le d} \widehat{f}(S) \cdot H_S(x). 
$$
Here $|S|$ denotes the $\ell_1$ norm of the vector $S$. We will define $\mathsf{W}^{\le d} [f] = \Vert f_{\le d} \Vert_2^2$ and $\mathsf{W}^{> d} [f] = \sum_{|S|>d} \Vert \widehat{f}(S) \Vert_2^2$. 
\subsubsection*{Ornstein-Uhlenbeck operator} 
\begin{definition}
The Ornstein-Uhlenbeck  operator $P_t$ is defined for $t \in [0, \infty)$ such that for any $f: \mathbb{R}^n \rightarrow \mathbb{R}^k$, 
$$
(P_t f)(x) = \int_{y \in \mathbb{R}^n} f(e^{-t} \cdot x + \sqrt{1- e^{-2t}} \cdot y) d \gamma_n(y).
$$
\end{definition}
Note that if $f : \mathbb{R}^n \rightarrow \Delta_k$, then so is $P_t f$ for every $t>0$. 
A basic fact about  the Ornstein-Uhlenbeck operator is that the functions $\{H_S\}$ are eigenfunctions of this operator. We leave the proof of the next proposition to the reader. 
\begin{proposition}
For $S \in \mathbb{Z}^{\ast n}$,
$P_t H_S = e^{-t \cdot |S|} \cdot H_S$. 
\end{proposition}

\subsubsection{Probabilistic inequalities}
\begin{theorem}\label{thm:hyper}
Let $p: \mathbb{R}^n \rightarrow \mathbb{R}$ be a degree-$d$ polynomial. Then, for any $t>0$, 
$$
\Pr_{x} \big[|p(x) - \mathbf{E}[p(x)]| \ge t \cdot \sqrt{\mathsf{Var}[p]}\big]  \leq d \cdot e^{-t^{2/d}}.
$$
\end{theorem}

\begin{theorem}\label{thm:combine-hyper}
Let $a, b: \mathbb{R}^n \rightarrow \mathbb{R}$ be degree $d$ polynomials satisfying $\mathbf{E}_x [a(x) - b(x)]=0$ and $\mathsf{Var}[a-b] \le (\tau/d)^{3d} \cdot \mathsf{Var}[a]$. Then, $\Pr_{x}[\mathsf{sign}(a(x)) \not = \mathsf{sign}(b(x))] = O(\tau)$. 
\end{theorem}

\subsubsection{Producing non-integral functions} 

Instead of producing functions $\{f^{(j)}_{\mathsf{junta}}\}_{1 \le i \le \ell}$ and $\{g^{(j)}_{\mathsf{junta}}\}_{1 \le i \le \ell}$ (in Theorem~\ref{thm:junta-strong}) with range $[k]$, we will actually produce functions $\{\tilde{f}^{(j)}_{\mathsf{junta}}\}_{1 \le i \le \ell}$ and $\{\tilde{g}^{(j)}_{\mathsf{junta}}\}_{1 \le i \le \ell}$ whose range will be close to $\Delta_{k,\epsilon}$. The next two lemmas show that functions with range $\Delta_{k,\epsilon}$  can be converted to non-interactive simulation strategies
with range $[k]$ with nearly the same guarantee. More precisely, we show that given $f', g' : \mathbb{R}^n \rightarrow \Delta_{k,\epsilon}$,  there are functions $f,g: \mathbb{R}^n \rightarrow [k]$ such that 
$\mathbf{E}[f] \approx \mathbf{E}[f']$,  $\mathbf{E}[g] \approx \mathbf{E}[g']$ and for any $1 \le j_1, j_2 \le k$, $\mathbf{E}[f_{j_1} P_t g_{j_2}] \approx \mathbf{E}[f'_{j_1} P_t g'_{j_2}]$. To define this, let us adopt the notation that given a point $x \in \mathbb{R}^k$, $\mathsf{Proj}(x)$ denotes the closest point  to $x$ in $\Delta_k$ in Euclidean distance. 

\begin{lemma}~\label{lem:round2}
Let $f  : \mathbb{R}^n \rightarrow \mathbb{R}^k$ which satisfies the following two conditions: 
\begin{enumerate}
\item  $\Pr_{x} [f(x) \not \in \Delta_{k,\delta}] \le \delta$. 
\item For all $x$, $\Vert f(x) \Vert_\infty \le k$. 
\end{enumerate}
Then, there is a function $f_1: \mathbb{R}^n \rightarrow \Delta_k$ such that $\Vert f-f_1 \Vert_1 = O(k \cdot \delta)$.
\end{lemma}
\begin{proof}
Define $f_1 = \mathsf{Proj}(f)$. Note that if $x$ is such that $f(x) \in \Delta_{k,\delta}$, then by definition, $\Vert f_1(x) -f(x) \Vert_1 \le \delta$. On the other hand, for any $x$, $\Vert f(x) -f_1(x) \Vert_1 \le k$.  This proves the claim. 
\end{proof}
\begin{lemma}~\label{lem:round3}
Let $f_1, g_1 : \mathbb{R}^n \rightarrow \Delta_k$. Then, there exist (explicitly defined) $f_2, g_2: \mathbb{R}^{n+2} \rightarrow [k]$ such that 
\begin{enumerate}
\item $\mathbf{E}[f_2 ] = \mathbf{E}[f_1]$ and  $\mathbf{E}[g_2 ] = \mathbf{E}[g_1]$.
\item For any $1 \le j, \ell \le k$, 
$$
\mathbf{E}[  f_{1,j} P_t g_{1,\ell}  ]=\mathbf{E}[  f_{2,j} P_t g_{2,\ell}].
$$
\end{enumerate}
{Further, the function $f_2$ (resp. $g_2$) is dependent only on $f_1$ (resp. $g_1$).}
\end{lemma}
\begin{proof}
Let $z = (x, z_1, z_2)$ where $x \in \mathbb{R}^n$ and $z_1, z_2 \in \mathbb{R}$. For any $y \in \Delta_k$, let us divide $\mathbb{R}$ into $k$ intervals $S_1, \ldots, S_k$ such that for $z \sim \gamma$, $\Pr[z \in S_i] = y_i$. For $y \in \Delta_k$ and $z' \in \mathbb{R}$, $\mathsf{Part}(y,z) = i$ if $z' \in S_i$. 
 Define $f_2: \mathbb{R}^{n+2} \rightarrow [k]$ as 
\[
f_2(z) =f_2(x,z_1,z_2) = \mathsf{Part}(f_1(x), z_1). 
\]
\[
g_2(z) =g_2(x,z_1,z_2) = \mathsf{Part}(g_1(x), z_2). 
\]
We will now verify the claimed properties. First of all, observe that the codomain of $f_2$ and $g_2$ is indeed $k$. Second, by definition, it is easy to follow that $\mathbf{E}[f_1]= \mathbf{E}[f_2]$ and $\mathbf{E}[g_1]= \mathbf{E}[g_2]$. Finally, note that 
$$
\mathbf{E}[  f_{1,j} P_t g_{1,\ell}  ]= \mathbf{E}_{(\bX^n, \bY^n) \sim \bP^n}[  f_{1,j}(\bX^n) g_{1,\ell}(\bY^n) ].
$$
On the other hand, suppose $z_1, z_2 \sim \gamma$. Then, 
$$
\Pr_{z_1, z_2 \sim \gamma} [f_2(x,z_1, z_2) = j \ \wedge \ g_2(y,z_1, z_2) = \ell] = f_{1,j}(x) g_{1,\ell}(y). 
$$
Thus, we obtain that 
$$
\mathbf{E}[  f_{2,j} P_t g_{2,\ell}  ] = \mathbf{E}_{(\bX^n, \bY^n) \sim \bP^n} [f_{1,j}(\bX^n) g_{1,\ell} (\bY^n)] = \mathbf{E}[  f_{1,j} P_t g_{1,\ell}  ]. 
$$
\end{proof}

\subsection{Proof strategy for the main theorem}

To describe the proof strategy for the main section, we first define a class of $k$-ary functions called \emph{polynomial plurality functions} (PPFs) which are closely related to the multivariate PTFs defined in the introduction but are somewhat different.  For this, let us first define the function $\arg \max$ as follows 
\begin{definition}
$\arg \max : \mathbb{R}^k \rightarrow \mathbb{R}^k$ is defined as \begin{equation*}
\arg \max (x_1, \ldots, x_k)   = \begin{cases} \mathbf{e}_i \ \ &\text{if } x_i > x_j \ \textrm{ for all } j\not =i \\ 
0 &\textrm{otherwise} \\ 
\end{cases}\end{equation*}
\end{definition}
\begin{definition}
A function $f: \mathbb{R}^n \rightarrow \mathbb{R}^k$ is said to be a PPF of degree-$d$ if there exists a polynomial $p: \mathbb{R}^n \rightarrow \mathbb{R}$ of degree $d$ and an index $1 \le j \le x$ such that 
$f= \arg \max(z)$ where $z_i = \delta_{i=j} \cdot p(x)$. 
Given polynomial $p: \mathbb{R}^n \rightarrow \mathbb{R}$ and $ 1 \le j \le k$, we define the function 
$\mathsf{PPF}_{p,  j}$ as 
\[
\mathsf{PPF}_{p,  j}(x) = \arg \max \big(\underbrace{0, \ldots, 0}_{(j-1) \textrm{ times}}, p(x) , \underbrace{0, \ldots, 0}_{(n-j) \textrm{ times}} \big). 
\]
\end{definition}
~\\The following is a basic fact about PPFs. 
\begin{fact}~\label{fact:balanced}
For any PPF $f$ of degree $d$, if $f=\mathsf{PPF}_{p,  j}$, we can assume without loss of generality
that $\mathsf{Var}(p)=1$. Further, by changing $f$ in at most $\delta$ fraction of places, we can assume that $|\mathbf{E}[p(x) ]|  \le d \cdot \log^{d/2}(1/\delta)$. Such a PPF is said to be a $(d,\delta)$-balanced PPF. 
\end{fact}
\begin{proof}
The fact about variance follows simply by scaling. To bound $|\mathbf{E}[p(x) ]|$, note that if $|\mathbf{E}[p(x)]| > d \cdot \log^{d/2}(1/\delta)$, then 
$\Pr_x [\mathsf{sign}(p(x)) = \mathsf{sign}(\mathbf{E}[p(x)])]  \ge 1-\delta$ (using Theorem~\ref{thm:combine-hyper}). Thus, if we set $q(x) = p(x) - \mathbf{E}[p(x)] + d \cdot \log^{d/2}(1/\delta) \cdot  \mathsf{sign}(\mathbf{E}[p(x)])$, then $\Pr_x [p(x) \not = q(x)] \le \delta$. The PPF defined as $\mathsf{PPF}_{q,j}$ satisfies all the desired properties. 
\end{proof}

To prove our main theorem (Theorem~\ref{thm:junta-strong}), we will prove the following two intermediate results. 
\begin{lemma}~\label{lem:smoothing}
For $1 \le i \le \ell$, let $f^{(i)}, g^{(i)}: \mathbb{R}^n \rightarrow [k]$ such that $\mathbf{E}[f^{(i)}] = \boldsymbol{\mu}^{(i)}_f$ and $\mathbf{E}[g^{(i)}]= \boldsymbol{\mu}^{(i)}_g$. Then, for any $t>0$, $\delta>0$, 
 $d_0 = d_0(t,k,\delta) = (2/t) \cdot \log(
k^2/\delta)$ and $1\le i \le \ell$, there are  
 functions $f^{(i)}_1, g^{(i)}_1: \mathbb{R}^n \rightarrow \mathbb{R}^k$ which satisfy the following conditions: 
\begin{enumerate}
\item For any $x \in \mathbb{R}^n$ and $1 \le i \le \ell$, $f^{(i)}_1(x), g^{(i)}_1(x)$ always lies in the positive orthant. 
\item  For any $x \in \mathbb{R}^n$ and $1 \le i \le \ell$, $\Vert f^{(i)}_1(x) \Vert_\infty , \Vert g^{(i)}_1(x) \Vert_\infty\le 1$. 
\item For $  1 \le i \le \ell$, $\Pr_{x} [f^{(i)}_1(x) \not \in \Delta_{k,k\delta/2}]\le\delta/2$ and $\Pr_{x} [g^{(i)}_1(x) \not \in \Delta_{k,k\delta/2}] \le \delta/2$.
\item  For $  1 \le i \le \ell$, $|\mathbf{E}[f^{(i)}_1] - \boldsymbol{\mu}^{(i)}_f|,\ |\mathbf{E}[g^{(i)}_1] - \boldsymbol{\mu}^{(i)}_g| =O(k \delta)$.
\item  For $  1 \le i,j  \le \ell$ and for any $1 \le s_1, s_2 \le k$, $|\mathbf{E}[ f^{(i)}_{1,s_1} P_t g^{(j)}_{1,s_2} ] -\mathbf{E}[ f^{(i)}_{s_1} P_t g^{(j)}_{s_2} ]| =O(k \cdot \delta)$.
\item For $1 \le i \le \ell$, $f^{(i)}_1$ and $g^{(i)}_1$ are of the following form. There are degree-$d_0$ polynomials $\{p^{(i)}_{s,j,1}\}_{1\le i \le \ell, 1\le s \le k, 1 \le j \le m}$ and  $\{p^{(i)}_{s,j,2}\}_{1\le i \le \ell, 1\le s \le k, 1 \le j \le m}$
\[
f^{(i)}_1 = \sum_{s=1}^k \sum_{j=1}^m \frac{1}{m} \cdot \mathsf{PPF}_{p^{(i)}_{s,j,1},j}(x) \ , \ g^{(i)}_1 = \sum_{s=1}^k \sum_{j=1}^m \frac{1}{m} \cdot \mathsf{PPF}_{p^{(i)}_{s,j,2},j}(x),
\]
such that the resulting PPFs $\mathsf{PPF}_{p^{(i)}_{s,j,1},j}(x)$ and $\mathsf{PPF}_{p^{(i)}_{s,j,2},j}(x)$ are $(d_0,\delta)$-balanced PPFs. Here $m= O(1/\delta)$. 
\end{enumerate}
{\red{Further, the function $f^{(i)}_1$ (resp. $g^{(i)}_1$)  is dependent only on $f^{(i)}$ (resp. $g^{(i)}$), $t$, $k$ and $\delta$.  }}
\end{lemma}

\begin{lemma}~\label{lem:junta-construction}
Let $\{p^{(i)}_{s,j,1}\}_{1 \le i \le \ell, 1\le s \le k, 1 \le j \le m}$ and  $\{p^{(i)}_{s,j,2}\}_{1 \le i \le \ell,1\le s \le k, 1 \le j \le m}$ be degree-$d_0$ polynomials. For $1 \le i \le \ell$, let 
$f^{(i)}_1, g^{(i)}_1: \mathbb{R}^n \rightarrow \mathbb{R}^k$  be defined as in Lemma~\ref{lem:smoothing} and satisfy the following two conditions: 
\begin{enumerate}
\item For $1\le i \le \ell$, $1 \le s \le k$ and $1 \le j \le m$, all the PPFs $\mathsf{PPF}_{p^{(i)}_{s,j,1},j}$ and $\mathsf{PPF}_{p^{(i)}_{s,j,2},j}$ are $(d_0,\delta)$-balanced PPFs. 
\item For $1\le i \le \ell$, $\Pr_{x} [f^{(i)}_1(x) \not \in \Delta_{k,\delta}] \le \delta$ and $\Pr_{x} [g^{(i)}_1(x) \not \in \Delta_{k,\delta}] \le \delta$.
\end{enumerate}
Then, there exists an explicit constant $n_0 = n_0 (d_0,k,\delta,\ell)$ such that there are polynomials $\{r_{s,j,1}^{(i)}\}_{1 \le i \le \ell, 1 \le s \le k, 1 \le j \le m}$ and $\{r_{s,j,2}^{(i)}\}_{1 \le i \le \ell,1 \le s \le k, 1 \le j \le m}$ satisfying the following conditions: For $1 \le i \le \ell$, let us define the functions $f^{(i)}_{\mathsf{junta}}, g^{(i)}_{\mathsf{junta}}: \mathbb{R}^{n_0} \rightarrow \mathbb{R}^k$ defined as 
\[
f^{(i)}_{\mathsf{junta}} = \sum_{s=1}^k \sum_{j=1}^m \frac{1}{m} \cdot \mathsf{PPF}_{r^{(i)}_{s,j,1},s}(x) \ , \ g^{(i)}_{\mathsf{junta}} = \sum_{s=1}^k \sum_{j=1}^m \frac{1}{m} \cdot \mathsf{PPF}_{r^{(i)}_{s,j,2},s}(x),
\]
Then, they satisfy the following three conditions: For all $1 \le i \le \ell$, 
\begin{enumerate}
\item $ \Vert \mathbf{E}[f^{(i)}_{\mathsf{junta}}]  - \mathbf{E}[f^{(i)}_1] \Vert_1 \le \delta$ and $ \Vert \mathbf{E}[g^{(i)}_{\mathsf{junta}}]  - \mathbf{E}[g^{(i)}_1] \Vert_1 \le \delta$.
\item $\Pr_{x} [f^{(i)}_{\mathsf{junta}}(x) \not \in \Delta_{k,\sqrt{\delta}}] \le \sqrt{\delta}$ and $\Pr_{x} [g^{(i)}_{\mathsf{junta}}(x) \not \in \Delta_{k,\sqrt{\delta}}] \le \sqrt{\delta}$.
\item For any $1 \le i, j \le \ell$, $1 \le s_1, s_2 \le k$, $|\mathbf{E}[ f^{(i)}_{1,s_1} P_t g^{(j)}_{1,s_2} ] -\mathbf{E}[ f^{(i)}_{\mathsf{junta},s_1} P_t g^{(j)}_{\mathsf{junta},s_2} ]| \le \delta$.
\end{enumerate}
\end{lemma}
{\red{\textbf{Proof of Theorem~\ref{thm:junta-strong}:}The proof of Theorem~\ref{thm:junta-strong} follows by applying Lemma~\ref{lem:smoothing} on the set $\{f^{(i)} \cup g^{(i)}\}_{1 \le i \le \ell}$ and subsequently applying Lemma~\ref{lem:junta-construction}. While the range of functions produced by $\{f_{\mathsf{junta}}^{(i)} \cup g_{\mathsf{junta}}^{(i)}\}_{1 \le i \le \ell}$ is not $\Delta_k$, by applying Lemma~\ref{lem:round2} and Lemma~\ref{lem:round3}, we can rectify this issue. We note here that the functions obtained in this process, namely $\{f_{\mathsf{junta}}^{(i)} \cup g_{\mathsf{junta}}^{(i)}\}_{1 \le i \le \ell}$ 
are explicit. Namely, the functions obtained before applying Lemma~\ref{lem:round2} and Lemma~\ref{lem:round3} are low-degree PPFs. Lemma~\ref{lem:round2} applies a projection on to the standard simplex $\Delta_k$. Likewise, Lemma~\ref{lem:round3} also produces an explicit function as its output. We now explain why $\{f_{\mathsf{junta}}^{(i)} \cup g_{\mathsf{junta}}^{(i)}\}_{1 \le i \le \ell}$ satisfy the stated guarantees.

In particular, overloading notation, let us denote the functions obtained by application of Lemma~\ref{lem:round2} and Lemma~\ref{lem:round3} as
$f_{\mathsf{junta}}^{(i)}$ and $ g_{\mathsf{junta}}^{(i)}$. Then, we see that \begin{equation}\Vert \mathbf{E}[f^{(i)}_{\mathsf{junta}}]  - \mathbf{E}[f^{(i)}_1] \Vert_1 \le O (k \cdot \sqrt{\delta}), \  \Vert \mathbf{E}[g^{(i)}_{\mathsf{junta}}]  - \mathbf{E}[g^{(i)}_1] \Vert_1 \le  O(k \cdot \sqrt{\delta}),  \end{equation}
\begin{equation}~\label{eq:tv-dist}
\textrm{For any }1 \le i, j \le \ell, \ 1 \le s_1, s_2 \le k, \ |\mathbf{E}[ f^{(i)}_{1,s_1} P_t g^{(j)}_{1,s_2} ] -\mathbf{E}[ f^{(i)}_{\mathsf{junta},s_1} P_t g^{(j)}_{\mathsf{junta},s_2} ]| \le \delta
\end{equation}
Note that the functions $\{f^{(i)} \cup g^{(i)}\}_{1 \le i \le \ell}$ have arity $n_0$. Further, observe that for $1 \le s_1, s_2 \le k$ and $1 \le i, j \le \ell$, 
\begin{eqnarray*}
\Pr [f_{\mathsf{junta}}^{(i)}(\bX^{n_0}) = s_1 \ \wedge \ g_{\mathsf{junta}}^{(j)}(\bY^{n_0}) = s_2] &=& \mathbf{E}[f^{(i)}_{\mathsf{junta},s_1} P_t g^{(j)}_{\mathsf{junta},s_2} ] \ \textrm{and}  \\
\Pr [f^{(i)}(\bX^{n}) = s_1 \ \wedge \ g^{(j)}(\bY^{n}) = s_2] &=& \mathbf{E}[f^{(i)}_{s_1} P_t g^{(j)}_{s_2} ]. 
\end{eqnarray*}
Thus, for $1\le s_1, s_2 \le k$, 
$$
\big| \Pr [f_{\mathsf{junta}}^{(i)}(\bX^{n_0}) = s_1 \ \wedge \ g_{\mathsf{junta}}^{(j)}(\bY^{n_0}) = s_2]- \Pr [f^{(i)}(\bX^{n}) = s_1 \ \wedge \ g^{(j)}(\bY^{n}) = s_2]\big| \le \delta. 
$$
This immediately implies that $$\dtv\big(\big(f_{\mathsf{junta}}^{(i)}(\bX^{n_0}) , g_{\mathsf{junta}}^{(j)}(\bY^{n_0})\big),\big(f_{}^{(i)}(\bX) , g_{}^{(j)}(\bY)\big)\big) = O(k^2 \delta),$$
which finishes the proof. 
}}

\section{Proof of Lemma~\ref{lem:smoothing}} 
The proof of Lemma~\ref{lem:smoothing} shall proceed in several steps. Note that Lemma~\ref{lem:smoothing}    claims existence of $\{f_1^{(i)}\}$ and $\{g_1^{(i)}\}$ which satisfies six different properties. The functions $\{f^{(i)}\}$ and $\{g^{(i)}\}$ themselves satisfy the first five properties and thus, the only non-trivial task that remains is to achieve the sixth property. The sixth property will be achieved by gradual modification of $\{f^{(i)}\}$ and $\{g^{(i)}\}$ in a sequence of steps which are explained below. 
\begin{enumerate}
\item Corollary~\ref{corr:Boosting} allows us to replace $f^{(i)}$ (resp. $g^{(i)}$) with $\fsm^{(i)}$ (resp. $\gsm^{(i)}$),
  which is the projection onto $\Delta_k$ of a polynomial, and which shares the same low-degree Hermite expansion as $f^{(i)}$ (resp. $g^{(i)}$).
Coupled with Claim~\ref{clm:noise-degree}, this shows that if $f^{(i)}$ is replaced by $\fsm^{(i)}$ and $g^{(i)}$ is replaced by $\gsm^{(i)}$, then the first five properties in Lemma~\ref{lem:smoothing} hold. On the other hand, note that while $\fsm^{(i)}$ and $\gsm^{(i)}$ do not have the full structure claim in Property~6, they do have some resemblance to PPFs.  Corollary~\ref{corr:Boosting} is the technically most innovative part of the proof and in turn relies on Lemma~\ref{lem:Boosting}. {\red{A crucial point for the application to non-interactive simulation is that the construction of $\fsm^{(i)}$ (resp. $\gsm^{(i)}$) is dependent only on $f^{(i)}$ (resp. $g^{(i)}$) and the error parameters. }}
\item Applying Bernstein-type approximations for Lipschitz functions in terms of low-degree polynomials, Lemma~\ref{lem:smoothing-1} shows that $\fsm^{(i)}$ and $\gsm^{(i)}$ can be replaced by $\fsm^{'(i)}$ and $\gsm^{'(i)}$ where each coordinate of $\fsm^{'(i)}$ and $\gsm^{'(i)}$ is a low-degree multivariate polynomial.
{\red{Again, crucially for the application to non-interactive simulation, the function $\fsm^{'(i)}$ (resp. $\gsm^{'(i)}$) is dependent only on $\fsm^{(i)}$ (resp. $\gsm^{(i)}$) and the error parameters.}} 
\item Finally, the functions $\fsm^{'(i)}$ and $\gsm^{'(i)}$ are changed to $f_1^{(i)}$ and $g_1^{(i)}$ which are linear combinations of PPFs (as promised in Lemma~\ref{lem:smoothing}) using some simple probabilistic observations. {\red{Again, the conversion of $\fsm^{'(i)}$  to $f_1^{(i)}$ is only dependent on   $\fsm^{'(i)}$ and desired error parameters. }}
\end{enumerate}

\subsection{Projections of polynomials}

We begin with the first step described above.
The first lemma relates the (by now, well-known) connection between the low-degree Hermite expansion of a function and its noise stability. In particular, it shows that if a pair of functions $(f^{(1)}, g^{(1)})$ (whose range is
$\Delta_k$) is replaced by another pair $(\underline{f}^{(1)}, \underline{g}^{(1)})$ such that low-degree Hermite spectrum of $f^{(1)}$ (resp. $g^{(1)}$) is close to that of $\underline{f}^{(1)}$ (resp. $\underline{g}^{(1)}$) are close to each other, then for any $1 \le s_1, s_2 \le k$, $\mathbf{E}[f^{(1)}_{s_1} P_t g^{(1)}_{s_2}] \approx \mathbf{E}[\underline{f}^{(1)}_{s_1} P_t \underline{g}^{(1)}_{s_2}]$.   
\begin{claim}\label{clm:noise-degree}
Let $f^{(1)}, g^{(1)} , \underline{f}^{(1)}, \underline{g}^{(1)} : \mathbb{R}^n \rightarrow \Delta_k$ such that for $d_1= d_1(\delta, t) = \frac{1}{t} \log (k^2/\delta)$ we have
$$\mathsf{W}^{\le d_1} [ (f^{(1)}-\underline{f}^{(1)}) ] ,  \ \mathsf{W}^{\le d_1} [g^{(1)}-\underline{g}^{(1)}]\le \delta^2/k^4.$$ Then, $\sum_{1 \le s_1, s_2 \le k} | \mathbf{E}[ f^{(1)}_{s_1} P_{t} g^{(1)}_{s_2} ]  - \mathbf{E}[ \underline{f}^{(1)}_{s_1} P_t \underline{g}^{(1)}_{s_2} ] | \le \delta$. 
\end{claim}
\begin{proof}
For any $1 \le s_1, s_2 \le k$, 
\begin{eqnarray*}
\big| \mathbf{E}[f^{(1)}_{s_1} P_t g^{(1)}_{s_2}]-\mathbf{E}[\underline{f}^{(1)}_{s_1} P_t \underline{g}^{(1)}_{s_2}]  \big| &\le& \big| \mathbf{E}[ (f^{(1)}_{s_1} - \underline{f}^{(1)}_{s_1})  P_t g^{(1)}_{s_2} ] \big| + \big| \mathbf{E}[ \underline{f}^{(1)}_{s_1} P_t (g^{(1)}_{s_2} - \underline{g}^{(1)}_{s_2}) ]  \big| 
\end{eqnarray*}
By using the self-adjointness of the noise operator and applying the Jensen's inequality, the first term can be bounded as
\[
\big| \mathbf{E}[ (f^{(1)}_{s_1} - \underline{f}^{(1)}_{s_1})  P_t g^{(1)}_{s_2} ] \big| \le \sqrt{\mathbf{E}[P_t(f^{(1)} - \underline{f}^{(1)})_{s_1}^2]}\sqrt{\mathbf{E}[(g^{(1)})_{s_2}^2]} \le \sqrt{\mathbf{E}[P_t(f^{(1)}- \underline{f}^{(1)})_{s_1}^2]}. 
\]
Similarly bounding $ \big| \mathbf{E}[ \underline{f}^{(1)}_{s_1} P_t (g^{(1)}_{s_2} - \underline{g}^{(1)}_{s_2}) ]  \big| $, we obtain 
\[
\big| \mathbf{E}[ (f^{(1)}_{s_1} - \underline{f}^{(1)}_{s_1})  P_t g^{(1)}_{s_2} ] \big| + \big| \mathbf{E}[ \underline{f}^{(1)}_{s_1} P_t (g^{(1)}_{s_2} - \tilde{g}^{(1)}_{s_2}) ]  \big| \le \sqrt{\mathbf{E}[P_t(f^{(1)}_{s_1} -    \underline{f}^{(1)}_{s_1})^2]} + \sqrt{\mathbf{E}[P_t(g^{(1)}_{s_2} - \underline{g}^{(1)}_{s_2})^2]}.
\]

Now, applying the condition that $\mathsf{W}^{\le d_1} [ (f^{(1)}-f^{(2)}) ] \le \delta^2/k^4$, we get that
$$ \mathbf{E}[\Vert P_t(f^{(1)} - \underline{f}^{(1)}) \Vert_2^2] \le \frac{\delta^2}{k^4} + e^{-2td_1} \cdot \mathbf{E}[\Vert (f^{(1)} - \underline{f}^{(1)}) \Vert_2^2] \le \frac{2 \delta^2}{k^4}.$$
The last inequality uses the fact that for all $x$, $\Vert f^{(1)}(x) - \underline{f}^{(1)}(x) \Vert_1 \le 1$. 
 Likewise, we also get $\mathbf{E}[\Vert P_t(g^{(1)} - \underline{g}^{(1)}) \Vert_2^2 \leq 2 \delta^2/k^4$. Combining this, we obtain that for all $1 \le s_1, s_2 \le k$, 
\[
 \big| \mathbf{E}[ f^{(1)}_{s_1} P_{t} g^{(1)}_{s_2} ]  - \mathbf{E}[ \underline{f}^{(1)}_{s_1} P_t \underline{g}^{(1)}_{s_2} ] \big|  \le \frac{2 \delta}{k^2}. 
\]
Summing over all $1 \le s_1, s_2 \le k$, we get the stated bound. 
\end{proof}
Next, we state the main technical lemma of this section. To state the lemma, we define the function $\mathsf{Proj}: \mathbb{R}^k \rightarrow \Delta_k$ such that $\mathsf{Proj}(x) = y$ if $y$ is the closest point (in Euclidean distance) to $x$ in $\Delta_k$. While the authors are aware 
that technically, we require $\mathsf{Proj}$ to be quantified by the parameter
$k$, the relevant $k$ shall always be clear from the context. 
\begin{lemma}~\label{lem:Boosting}
Let $F: \mathbb{R}^n \rightarrow \Delta_k$ and let $g_1, \ldots, g_m : \mathbb{R}^n \rightarrow \mathbb{R}^k$ be an orthonormal sequence of functions under the standard $n$-dimensional Gaussian measure $\gamma_n$.
Here the function $g_1 : x \mapsto (1,\ldots, 1)$.  
 Then, for any $\delta>0$, there exists a function $F_{\mathsf{proj}} : \mathbb{R}^n \rightarrow \Delta_k$ of the form 
$F_{\mathsf{proj}} = \mathsf{Proj} (\sum_{i=1}^m \kappa_i g_i)$ satisfying
\[
  \sum_{i=1}^m (\E[g_i F] - \E[g_i F_{\mathsf{proj}}])^2 \le \delta.
\]
Further, $\sum_{i=1}^m \Vert \kappa_i \Vert_2^2 \le \delta^{-2}$. 
\end{lemma}
Before proving Lemma~\ref{lem:Boosting}, we first see why this lemma is useful. In particular, we have the following corollary. Essentially, the corollary says that given $f, g: \mathbb{R}^n \rightarrow \Delta_k$, there are functions $f_{\mathsf{sm}}$ and $g_{\mathsf{sm}}$ such that (i) the low-level Hermite spectrum of $f$  (resp. $g$) is close to $f_{\mathsf{sm}}$ (resp. $g_{\mathsf{sm}}$) (ii) Both $\fsm$ and $\gsm$ are obtained by applying the function $\mathsf{Proj}$ on a low-degree polynomial. In essence, we are obtaining \emph{simple} functions $\fsm$ and $\gsm$ which simultaneously (i) have the same low-level Hermite spectrum as $f$ and $g$ (ii) and have range $\Delta_k$. 
\begin{corollary}\label{corr:Boosting}
Given function $f: \mathbb{R}^n \rightarrow [k]$, $d \in \mathbb{N}$ and error parameter $\delta>0$,  there is a function $\fsm: \mathbb{R}^n \rightarrow \Delta_k$ which has the following properties: 
\begin{enumerate}
\item The function $\fsm$ has the following form: $$\fsm(x) = \mathsf{Proj}\bigg(\sum_{|S| \le d} \alpha_{f,s} H_S(x)\bigg) ,$$ where $H_S(x)$ is the Hermite polynomial corresponding to the multiset $S$. 
\item $\sum_{|S| \le d} \Vert \alpha_{f,S} \Vert_2^2 \le \delta^{-2}$. 
\item Define $\beta_{f,S} = \mathbf{E}[\fsm (x) \cdot H_S(x)]$. Then, $\sum_{|S| \le d}  \Vert \beta_{f,S} - \alpha_{f,S} \Vert_2^2 \le \delta$. 
\end{enumerate}
We note that for a scalar-valued function $H_S$ and a vector-valued function $\fsm$, we compute $\mathbf{E}[\fsm \cdot H_S]$ pointwise for each coordinate of the vector valued function $\fsm$. 
\end{corollary}
The proof of this corollary follows straightaway by instantiating Lemma~\ref{lem:Boosting} with $\{g_1, \ldots, g_m\} = \{H_{S}\}_{|S| \le d}$ with $F=f$ and 
$F=g$. 
\begin{proofof}{Lemma~\ref{lem:Boosting}}
We will prove this lemma via an iterative argument.
We will define a sequence of functions $\{F_t\}_{t \ge 0}$ iteratively such that for all $t \ge 0$, $F_t : \mathbb{R}^n \rightarrow \Delta_k$. 
Define the vector $\beta \in \mathbb{R}^m$ by
$\beta_j = \langle F, g_j \rangle$.  
Also,  for every $t \ge 0$, we will define $\beta_t\in \mathbb{R}^m$ by $\beta_{t,j} = \langle F_t, g_j \rangle$. 
The iterative process has the following property:
If for any $t$, $\Vert \beta_t - \beta \Vert_2^2 \le \delta$, then we terminate the process. Else, we modify $F_t$ to obtain the function $F_{t+1}$. We now define the initial function $F_0$ as well as the modification to obtain $F_{t+1}$ from $F_t$ (when $t \ge 0$). 

 The function $F_0 : \mathbb{R}^n \rightarrow \Delta_k$ is defined as $F_0 : x \rightarrow (1/k, \ldots, 1/k)$. 
Next, given $F_t$, we define $F_{t+1}$.
To do this, 
 we will also need to define an auxiliary sequence of functions 
$\{G_t \}_{t \ge 0}$ where $G_0 = F_0$.
The iterative process is defined in Figure~\ref{fig:DS}. 
\begin{figure}[htb]
\hrule
\vline
\begin{minipage}[t]{0.98\linewidth}
\vspace{10 pt}
\begin{center}
\begin{minipage}[h]{0.95\linewidth}


\vspace{3 pt}
\underline{\textsf{Description of iterative process }}
\begin{enumerate}
\item Define $\rho_t = \Vert \beta_t - \beta \Vert_2$. \item If $\rho_t^2  \le  \delta$, then stop the process. Else, we define $J_t = \sum_{j=1}^m (\beta -\beta_t)_j \cdot g_j$.  
\item Define $G_{t+1} = G_t+ J_t/2$. Define $F_{t+1}  = \mathsf{Proj}(G_{t+1})$ and $t \leftarrow t+1$. Go to Step 1. 
\end{enumerate}


\vspace{5 pt}

\end{minipage}
\end{center}

\end{minipage}
\hfill \vline
\hrule
\caption{Iterative process describing the sequence $\{F_t\}$}
\label{fig:DS}
\end{figure}
~\\
It is clear that if this process terminates at step $t=t_0$, then the function $F_{\mathsf{proj}}= F_{t_0}$ satisfies the required properties. Thus, we now need to bound the convergence rate of the process. To do this, we introduce a potential function 
$\Psi(t)$ defined as follows: 
\[
\Psi(t) = \mathbf{E}[ \langle F - F_t , F - 2 G_t + F_t \rangle].  
\]
The basic observation here is that $\Psi(0) = O(1)$. We will prove two main lemmas. The first will prove that in every iteration of the process in Figure~\ref{fig:DS}, $\Psi(t)$ decreases by a fixed amount. The second is that $\Psi(t)$ is always non-negative. These two facts, in conjunction, automatically imply an upper bound on the maximum number of steps in the algorithm. 
\begin{claim}~\label{clm:potential-decrease}
$$
\mathbf{E}[\langle F - F_t, J_t \rangle] = \rho_t^2. 
$$
\end{claim}
\begin{proof}
By orthogonality of the functions $\{g_j \}_{j=1}^m$, 
$$
\mathbf{E}[\langle F - F_t, J_t \rangle]  = \sum_{j=1}^m (\beta -\beta_t)_j   \mathbf{E}[\langle g_j , F- F_t \rangle]  = \sum_{j=1}^m (\beta -\beta_t)_j  \cdot (\beta-\beta_t)_j = \Vert \beta - \beta_t \Vert_2^2. 
$$
\end{proof}

We now recall a basic fact about projective maps (see, e.g.~\cite[Theorem 3]{CheneyGoldstein:59}).
\begin{fact}~\label{fact:convex}
Let $C$ be a closed, convex set and let $\mathsf{Proj}_C : \mathbb{R}^n \rightarrow C$ be defined as $x \mapsto \arg \min_{y \in C} \Vert x - y \Vert_2$. Then the map $\mathsf{Proj}_C$ is uniquely defined, and always contractive i.e. for any $z, z' \in \mathbb{R}^n$, $\Vert \mathsf{Proj}_C (z) - \mathsf{Proj}_{C}(z') \Vert_2 \le  \Vert z- z' \Vert_2$. 
Moreover, for any $x \in C$ and any $z \in \R^n$, $\langle z - \mathsf{Proj}_C(z), x - \mathsf{Proj}_C(z)\rangle \le 0$.
\end{fact} 

\begin{claim}~\label{clm:non-negative}
For all $t$, $\Psi(t) \geq 0$. 
\end{claim}
\begin{proof}
\begin{eqnarray*}
\Psi(t) &=&  \mathbf{E}[ \langle F - F_t , F - 2 G_t + F_t \rangle]  \\
&=& \mathbf{E}[ \langle F - F_t , F -  F_t \rangle] + 2 \cdot \mathbf{E}[ \langle F - F_t , F_t -  G_t \rangle] . 
\end{eqnarray*}
The first term is clearly non-negative. The second is non-negative by Fact~\ref{fact:convex}, taking
$z = G_t$ and $x = F$.
\end{proof}

The next lemma shows that the potential function always decreases by a fixed quantity. 
\begin{lemma}~\label{lem:descent} 
$$
\Psi(t+1) - \Psi(t) \le -\frac{\rho_t^2}{4}. 
$$
\end{lemma}
\begin{proof}
\begin{eqnarray*}
\Psi(t+1) - \Psi(t) &=&  \mathbf{E}[ \langle F - F_{t+1} , F - 2 G_{t+1} + F_{t+1} \rangle]-  \mathbf{E}[ \langle F - F_t , F - 2 G_t + F_t \rangle] \\
&=& \mathbf{E}[\langle F  - F_t, 2 (G_t - G_{t+1}) \rangle] + \mathbf{E}[\langle F_{t+1} - F_t , 2 G_{t+1} - F_t - F_{t+1} \rangle] \\
&=& \mathbf{E}[\langle F - F_t, - J_t \rangle]  +  \mathbf{E}[\langle F_{t+1} - F_t, 2 G_{t+1} - F_t - F_{t+1} \rangle] \\
&=& - \rho_t^2 +  \mathbf{E}[\langle F_{t+1} - F_t , 2 G_{t+1} - F_t - F_{t+1} \rangle] ~\textrm{(applying Claim~\ref{clm:potential-decrease})} \\
&=& -\rho_t^2 +  2 \cdot \mathbf{E}[\langle F_{t+1} - F_t ,  G_{t+1} -  F_{t+1} \rangle] +  \mathbf{E}[\langle F_{t+1} - F_t ,  F_{t+1} -  F_{t} \rangle] \\ 
&=& -\rho_t^2 + \mathbf{E}[\Vert F_{t+1} - F_{t} \Vert_2^2] + 2 \cdot \mathbf{E}[\langle F_{t+1} - F_t ,  G_{t+1} -  F_{t+1} \rangle] \\
&\le& -\rho_t^2 + \mathbf{E}[ \Vert G_{t+1} - G_{t} \Vert_2^2 ] +  2 \cdot \mathbf{E}[\langle F_{t+1} - F_t , G_{t+1} -  F_{t+1} \rangle]~ \textrm{(applying Fact~\ref{fact:convex})} \\
&=& - \frac{3 \rho_t^2}{4} + 2 \cdot \mathbf{E}[\langle F_{t+1} - F_t ,  G_{t+1} -  F_{t+1} \rangle]
\end{eqnarray*}

It remains to show that $\mathbf{E}[\langle F_{t+1} - F_t ,  G_{t+1} -  F_{t+1} \rangle] \le \frac{\rho_t^2}{4}$.
Indeed, the Cauchy-Schwarz inequality yields
\begin{align*}
  \|F_{t+1} - F_t\|_2 \|G_{t+1} - G_t\|_2
  &\ge \langle G_{t+1} - G_t, F_{t+1} - F_t \rangle \\
  &= \langle G_{t+1} - F_{t+1}, F_{t+1} - F_t \rangle + \langle F_{t+1} - F_t, F_{t+1} - F_t\rangle + \langle F_t - G_t, F_{t+1} - F_t\rangle
\end{align*}
In the last line above, the second term is obviously non-negative. Moreover, the third term is non-negative by Fact~\ref{fact:convex} (take $z = G_t$ and $x = F_{t+1}$). Hence,
\[
  \langle G_{t+1} - F_{t+1}, F_{t+1} - F_t \rangle
  \le \|F_{t+1} - F_t\|_2 \|G_{t+1} - G_t\|_2
  \le \|G_{t+1} - G_t\|_2^2 = \frac{\rho_t^2}{4},
\]
where the second inequality follows from Fact~\ref{fact:convex}.

\end{proof}
Combining Claim~\ref{clm:non-negative} and Lemma~\ref{lem:descent}, we obtain that the iterative process described in Figure~\ref{fig:DS} stops in at most $4/\delta$ steps. If the above iteration stops after $t=t_0$ steps, we let $F_{\mathsf{proj}}= F_{t_0}$.
Note that $ F_{\mathsf{proj}} = \mathsf{Proj}(\sum_{0 \le t < t_0} J_t/2)$. Thus, it is clear that $F_{\mathsf{proj}} = \mathsf{Proj}(\sum_{i=1}^m \kappa_i g_i)$. 
To bound $\sum_{i=1}^m \Vert \kappa_i \Vert_2^2$, note that $$\sum_{i=1}^m \Vert \kappa_i \Vert_2^2  = \Vert \sum_{0 \le t < t_0} J_t/2 \Vert_2^2 \le t_0 \cdot \sum_{0 \le t < t_0} \Vert J_t/2\Vert_2^2 \le t_0^2 \cdot \max_t \Vert J_t/2 \Vert_2^2 \le t_0^2.$$
The very last inequality uses the fact that $\Vert J_t \Vert_2 \le \Vert (F_t - F) \Vert_2 \le 1$. Plugging the upper bound of $O(1/\delta^2)$ on $t_0^2$, we obtain that $\sum_{i=1}^m \Vert \kappa_i \Vert_2^2 \le O(1/\delta^2)$. This concludes the proof. 
\end{proofof}

\begin{corollary}~\label{corr:fsm}
For $t>0$, error parameter $\delta>0$ and any function $f: \mathbb{R}^n \rightarrow [k]$, there is a function $\fsm: \mathbb{R}^n \rightarrow \Delta_k$ such that for $d =(2/t) \cdot \log(k^2/\delta)$, we have the following: 
\begin{enumerate}
\item $\Vert \mathbf{E}[\fsm]  -\mathbf{E}[f] \Vert_1 \le \delta$.
\item The function $\fsm=\mathsf{Proj}(p_{f,1}(x), \ldots, p_{f,k}(x))$ where for all $1 \le s \le k$, $p_{f,s} : \mathbb{R}^n \rightarrow \mathbb{R}$ are polynomials of degree $d$ and $\mathsf{Var}(p_{f,s}) \le k^8/\delta^4$. 
\item For any $g: \mathbb{R}^n \rightarrow [k]$ and  the corresponding function $\gsm: \mathbb{R}^n \rightarrow \Delta_k$, we have 
$\sum_{1 \le s_1, s_2 \le k} |\mathbf{E}[f_{\mathsf{sm},s_1} P_t g_{\mathsf{sm},s_2}]-\mathbf{E}[f_{s_1} P_t g_{s_2}]| \leq \delta$.
\end{enumerate}
\end{corollary}
\begin{proof}
Given the function $f: \mathbb{R}^n \rightarrow [k]$, 
we apply Corollary~\ref{corr:Boosting} to obtainthe  function $\fsm: \mathbb{R}^n \rightarrow \Delta_k$ where
\[
\fsm=\mathsf{Proj}(p_{f,1}(x), \ldots, p_{f,k}(x)),
\]
where for all $1 \le s \le k$, $p_{f,s}: \mathbb{R}^n \rightarrow \mathbb{R}$ are polynomials of degree 
$d = (1/t) \cdot \log(k^2/\delta)$ such that $\mathsf{W}^{\le d} [(\fsm - f)] \le \delta^2/k^4$. 
Further, for each $1 \le s \le k$, $\mathsf{Var}(p_{f,s}) \le (k^8/\delta^4)$. This immediately implies both items 1 and 2. To prove Item 3, note that we also have $\mathsf{W}^{\le d} [(\gsm - g)] \le \delta^2/k^4$.  
Applying Claim~\ref{clm:noise-degree}, we obtain that 
$\sum_{1 \le s_1, s_2 \le k} |\mathbf{E}[f_{\mathsf{sm},s_1} P_t g_{\mathsf{sm},s_2}]-\mathbf{E}[f_{s_1} P_t g_{s_2}]| \leq \delta$. This proves Item 3. 
\end{proof}
This completes the first step in the outline of Lemma 5: we have replaced arbitrary functions by projections of polynomials.

\subsection{Bernstein approximation}
The next step in the proof of Lemma 5 is the removal of the projection. The basic idea is just to approximate
the projection map by a polynomial. Then, the projection of a polynomial becomes the composition of two polynomials,
which is still a polynomial.

\begin{definition}
  For $0 \le k \le d$, efine $p_{k,d}(x) = \binom{d}{k} x^k (1-x)^{d-k}$. For a function $f: [0, 1]^\ell \to \mathbb{R}$,
  define the polynomial $\mathsf{BP}_{f, d_1, \dots, d_\ell}$ by
  \[
    \mathsf{BP}_{f, d_1, \dots, d_\ell}(x)
    = \sum_{k_1, \dots, k_\ell} f\left(\frac{k_1}{d_1}, \dots, \frac{k_\ell}{d_\ell}\right) p_{k_1,d_1}(x_1) \cdots p_{k_\ell,d_\ell}(x_\ell).
  \]
  We call $\mathsf{BP}_{f,d_1, \dots, d_\ell}$ the multivariate Bernstein approximation for $f$ with
  degrees $(d_1, \dots, d_\ell)$.
\end{definition}
\begin{theorem}\label{thm:Bernstein}\textbf{Multivariate Bernstein approximations}
Let $f: [0,1]^\ell \rightarrow \mathbb{R}$ be a $L$-Lipschitz function  in $[0,1]^{\ell}$. In other words, 
$\Vert f(x) - f(y) \Vert_2 \le L \cdot \Vert x - y \Vert_2$.
Then $\mathsf{BP}_{f, d_1, \ldots, d_\ell}$ satisfies the inequality
\[
  \sup_{z \in [0,1]^{\ell}} \big| f(z) - \mathsf{BP}_{f, d_1, \ldots, d_\ell}(z) \big| \le \frac{L}{2} \cdot \bigg( \sum_{j=1}^\ell \frac{1}{d_j}\bigg)^{1/2}
\]
\end{theorem}

The proof of Theorem~\ref{thm:Bernstein} is folklore; we provide one for completeness.
\begin{proof}
  Fix $z \in [0, 1]^\ell$. Note that each $p_{k_i,d_i}(z_i)$ is non-negative, and that $\sum_{k_i=0}^{d_i} p_{k_i,d_i}(z_i) = 1$.
  Hence,
  \begin{align*}
    f(z) - \mathsf{BP}_{f,d_1, \dots, d_\ell}(z)
    & = \sum_{k_1, \dots, k_\ell} \left[f(z) - f\left(\frac{k_1}{d_1}, \dots, \frac{k_\ell}{d_\ell}\right)\right] p_{k_1,d_1}(z_1) \cdots p_{k_\ell,d_\ell}(z_\ell) \\
    & \le L \sum_{k_1, \dots, k_\ell} \left\|z - \left(\frac{k_1}{d_1}, \dots, \frac{k_\ell}{d_\ell}\right)\right\|_2 p_{k_1,d_1}(z_1) \cdots p_{k_\ell,d_\ell}(z_\ell) \\
    & \le L \left[\sum_{k_1, \dots, k_\ell} \left\|z - \left(\frac{k_1}{d_1}, \dots, \frac{k_\ell}{d_\ell}\right)\right\|_2^2 p_{k_1,d_1}(z_1) \cdots p_{k_\ell,d_\ell}(z_\ell)\right]^{1/2} \\
  & = L \left[\sum_{i=1}^\ell \sum_{k_i=0}^{d_i} \Big(z_i - \frac{k_i}{d_i}\Big)^2 p_{k_i,d_i}(z_i)\right]^{1/2}.
  \end{align*}
  Finally, note that $\sum_{k=0}^d (x - k/d)^2 p_{k,d}(x)$ is just the variance of a binomial random variable
  with $d$ trials and success probability $x$. This is bounded by $\frac{1}{4d}$. Plugging in this bound for each $i$ separately
  completes the proof.
\end{proof}

Rescaling the function, we have the following corollary. To state this corollary, we let 
$B(x, r) = \{z : \Vert z - x \Vert_2 \le r\}$ i.e. the $\ell_2$ of radius $r$ at $x$. 
\begin{corollary}\label{corr:Bernstein}
Let $f: B(x,r) \rightarrow \mathbb{R}$ be a $1$-Lipschitz function (where $B(x,r) \subseteq \mathbb{R}^\ell$). Then, given any error parameter $\eta>0$, there is a polynomial $p_{f, r, \eta}$ whose degree in every variable is at most  $d_{B}(\eta, r,\ell) = \ell \cdot 4r^2 \cdot (1/\eta^2) 
$ such that 
\[
\sup_{z \in B(x,r)} \big|p_{f, r, \eta}(z) - f(z) \big| \le \eta. 
\]
\end{corollary}
\begin{proof}
To prove this, we will rely on Theorem~\ref{thm:Bernstein}. First, define $B_{\infty}(x,r) =\{z: \Vert z-x\Vert_\infty \le r\}$. We extend $f$ to $B_{\infty}(x,r)$ as follows: 
$
f(z) = f(\mathsf{Proj}_{B(x,r)}(z)). 
$ Note that the extension is $1$-Lipschitz (using Fact~\ref{fact:convex}). Define the function $g: [0,1]^\ell \rightarrow \mathbb{R}$ as 
\[
g(z) = f\bigg(x + \bigg(z-\mathbf{\frac{1}{2}}\bigg) \cdot 2r \bigg). 
\]
Here $\mathbf{\frac{1}{2}}$ is the point in $\mathbb{R}^{\ell}$ which is $1/2$ in every coordinate. It is easy to see that the function $g$ is $2r$-Lipschitz. Thus, if we choose the function $\mathsf{BP}_{g,d_1, \ldots, d_{\ell}}$, then we have
\[
\sup_{z \in [0,1]^{\ell}} \big|\mathsf{BP}_{g,d_1, \ldots, d_{\ell}} - g(z) \big| \le 2r \cdot \bigg(\sum_{j=1}^{\ell} \frac{1}{d_j} \bigg)^{1/2}.
\]
In particular, we set all the degrees $d_1 = \ldots =d_{\ell}  = \ell \cdot 4r^2 \cdot (1/\eta^2)$, then $\sup_{z \in [0,1]^{\ell}} \big|\mathsf{BP}_{g,d_1, \ldots, d_{\ell}} - g(z) \big| \le \eta$. Thus, if we set $p_{f,r,\eta}(z)$ as 
\[
p_{f,r,\eta}(z) = \mathsf{BP}_{g,d_1, \ldots, d_{\ell}} \bigg(  \frac{z-x}{2r} + \mathbf{\frac{1}{2}}\bigg).
\]
It is clear that the polynomial $p_{f,r,z}$ satisfies $\sup_{z \in B(x,r)} \big|p_{f, r, \eta}(z) - f(z) \big| \le \eta$.  
\end{proof}
We next modify the function $\fsm: \mathbb{R}^n \rightarrow \Delta_k$ obtained in Corollary~\ref{corr:fsm} to obtain the function $\fsm' : \mathbb{R}^n \rightarrow \mathbb{R}^k$ which is a (i) low-degree polynomial and (ii) $\fsm$  is close to $\fsm'$  with high probability on the Gaussian measure $\gamma_n$. 
\begin{lemma}\label{lem:smoothing-1}
Given  the function $\fsm: \mathbb{R}^n \rightarrow \Delta_k$ from Corollary~\ref{corr:fsm}, there is a  function $\fsm': \mathbb{R}^n \rightarrow \mathbb{R}^k$ such that $\fsm' = (p'_{f,1}(x), \ldots, p'_{f,k}(x))$  where for all $1 \le s \le k$, $p'_{f,s}: \mathbb{R}^n \rightarrow \mathbb{R}$ are polynomials satisfying the following conditions:
\begin{enumerate}
\item For $1 \le s \le k$, the polynomials $\{p'_{f,s}\}$  have degree $d' = \log^{d}(dk/\delta) \cdot \mathsf{poly}(k/\delta) \cdot d$ where $d$ is the degree appearing in Corollary~\ref{corr:fsm}.  
\item $\Pr_{x \sim \gamma_n} [\Vert \fsm(x) - \fsm'(x) \Vert_\infty \le \delta/4] \le \delta/2$. 
\end{enumerate}
\end{lemma}
\begin{proof}
Let the function $\fsm(x) = \mathsf{Proj}(p_{f,1}(x), \ldots, p_{f,k}(x))$. Since all the polynomials are degree $d$ and have variance at most $\sigma_{\mathsf{sm}}^2 = k^{8}/\delta^4$, using Theorem~\ref{thm:hyper}, we obtain the following: 
\begin{equation}\label{eq:inf-ball-1}
\Pr_{x \sim \gamma_n} \sup_{1 \le s \le k} [|p_{f,s} - \mathbf{E}[p_{f,s}]| \le \log^{d/2} (2dk/\delta) \cdot \sigma_{\mathsf{sm}}] \le \frac{\delta}{2}. 
\end{equation}
Define the point $\boldsymbol{\mu}_{sm,f} = (\mathbf{E}[p_{f,1}], \ldots, \mathbf{E}[p_{f,s}])$.  Also, let $r_{sm} = \log^{d/2} (2dk/\delta) \cdot \sigma_{\mathsf{sm}}$. Since the projection from $\R^k$ to $\Delta_k$ is Lipschitz, Corollary~\ref{corr:Bernstein} implies that there exist polynomials $p_{\mathsf{sm},s}: \mathbb{R}^{k} \rightarrow \mathbb{R}$  (for $1 \le s\le k$) whose degree in every variable is at most $k \cdot 4 r_{sm}^2 \cdot 16/\delta^2 = \log^d(dk/\delta) \cdot \mathsf{poly}(k/\delta)$, and which satisfy
\begin{equation}~\label{eq:inf-ball} 
\textrm{for all } z \in B(\boldsymbol{\mu}_{sm,f}, r_{sm}) \textrm{, we have } \ |p_{\mathsf{sm},s}(z) - \mathsf{Proj}_s(z)| \le \frac{\delta}{4}
 \end{equation}
 Let $p_{\mathsf{sm}}: \mathbb{R}^k \rightarrow \mathbb{R}^k$ be defined as the map $p_{\mathsf{sm}}(x) = (p_{\mathsf{sm},1}(x), \ldots, p_{\mathsf{sm},k}(x))$. 
 Recall that $\fsm = \mathsf{Proj} (p_{f,1}(x), \ldots, p_{f,k}(x))$. We define $p'_f= p_{\mathsf{sm}} \circ (p_{f,1}, \ldots, p_{f,k})$. 
We now define $\fsm' = (p'_{f,1}(x), \ldots, p'_{f,k}(x))$. It is clear that for $1 \le s \le k$, $p'_{f,s}$ is a polynomials of degree $\log^{d}(dk/\delta) \cdot \mathsf{poly}(k/\delta) \cdot d$. Likewise, combining (\ref{eq:inf-ball}) and (\ref{eq:inf-ball-1}), we obtain that  $\Pr_{x \sim \gamma_n} [\Vert \fsm(x) - \fsm'(x) \Vert_\infty \le \delta/2] \le \delta/2$. 
\end{proof}

\subsection{Converting to PPFs}
Before we finish the proof of Lemma~\ref{lem:smoothing}, we will need to 
make a couple of elementary observations. First of all, 
observe that if $\alpha$ is uniformly random in $[0,1]$, then for any $x \in [0,1]$, $\mathbf{E}[\mathbf{1}_{x - \alpha \ge 0}] = x$. Here $\mathbf{1}_{x-\alpha \ge 0}$ denotes the function which is $1$ if $x-\alpha \ge 0$ and $0$ otherwise.  Now, for any parameter $\eta>0$, define the distribution $\mathsf{Int}_\eta$ to be uniformly random over the set $\{i \cdot \eta\}_{i \ge 0} \cap [0,1]$. Then, we have the following simple claim. 
\begin{claim}~\label{clm:expectation-delta}
Let $\zeta>0$ and $y \in \Delta_{k,\zeta}$. Then, 
\[
\bigg\Vert \mathop{\mathbf{E}}_{(\alpha_1,\ldots, \alpha_k) \sim \mathsf{Int}_\eta^k} \bigg[\sum_{s=1}^k \arg \max (\underbrace{0, \ldots, 0}_{s-1 \textrm{ times}}, y_s-\alpha_s,\underbrace{0, \ldots, 0}_{k-s \textrm{ times}} )\bigg]  - y \bigg\Vert_1 \le 2( \zeta + k \cdot \eta).
\]
\end{claim}
\begin{proof}
Let the point closest to $y$ in $\Delta_k$ be $x$. Then, we have $\Vert x-y \Vert_1=\zeta$. We have the following:
\[
\bigg\Vert \mathop{\mathbf{E}}_{(\alpha_1,\ldots, \alpha_k) \sim \mathsf{Int}_\eta^k} \bigg[\sum_{s=1}^k \arg \max (\underbrace{0, \ldots, 0}_{s-1 \textrm{ times}}, x_s-\alpha_s,\underbrace{0, \ldots, 0}_{k-s \textrm{ times}} )\bigg]  - x \bigg\Vert_1 \le  k \cdot \eta.
\]
Combining this with $\Vert x-y \Vert_1\le \zeta$, we obtain 
\begin{equation}~\label{eq:inter}
\bigg\Vert \mathop{\mathbf{E}}_{(\alpha_1,\ldots, \alpha_k) \sim \mathsf{Int}_\eta^k} \bigg[\sum_{s=1}^k \arg \max (\underbrace{0, \ldots, 0}_{s-1 \textrm{ times}}, x_s-\alpha_s,\underbrace{0, \ldots, 0}_{k-s \textrm{ times}} )\bigg]  - y \bigg\Vert_1 \le  k \cdot \eta +\zeta.
\end{equation}
Next, for any $1 \le s\le k$, 
$$
\Vert \mathop{\mathbf{E}}_{(\alpha_1,\ldots, \alpha_k) \sim \mathsf{Int}_\eta^k} \arg \max (\underbrace{0, \ldots, 0}_{s-1 \textrm{ times}}, x_s-\alpha_s,\underbrace{0, \ldots, 0}_{k-s \textrm{ times}} ) - \arg \max (\underbrace{0, \ldots, 0}_{s-1 \textrm{ times}}, y_s-\alpha_s,\underbrace{0, \ldots, 0}_{k-s \textrm{ times}} ) \Vert_1 \le |x_s-y_s| + \eta. 
$$
Summing over all $1 \le s \le k$ and combining with (\ref{eq:inter}), we obtain the claim. 
\end{proof}
\begin{proofof}{Lemma~\ref{lem:smoothing}}
For $1 \le i \le \ell$, let $\{\fsm^{'(i)}\}$  and $\{\gsm^{'(i)}\}$ be the functions obtained by applying Corollary~\ref{corr:fsm} and Lemma~\ref{lem:smoothing-1} to the family of functions $\{f^{(i)}\}$ and $\{g^{(i)}\}$. 
 In particular, let $\fsm^{'(i)}  = (p^{'(i)}_{f,1}, \ldots, p^{'(i)}_{f,k})$ and $\gsm^{'(i)} = (p^{'(i)}_{g,1}, \ldots, p^{'(i)}_{g,k})$. 
For $ \eta>0$ (to be fixed later), let us define $f^{(i)}_1$ and $g^{(i)}_1$ as follows: 
\[
f^{(i)}_1 = \sum_{s=1}^k \mathop{\mathbf{E}}_{(\alpha_1, \ldots, \alpha_k) \in \mathsf{Int}_{\eta}^k} \arg \max \big(\underbrace{0, \ldots, 0}_{s-1 \ \textrm{times}} , p^{'(i)}_{f,s} - \alpha_s, \underbrace{0, \ldots, 0}_{k-s \ \textrm{times}}  \big)
\]
\[
g^{(i)}_1 = \sum_{s=1}^k \mathop{\mathbf{E}}_{(\alpha_1, \ldots, \alpha_k) \in \mathsf{Int}_{\eta}^k} \arg \max \big(\underbrace{0, \ldots, 0}_{s-1 \ \textrm{times}} , p^{'(i)}_{g,s} - \alpha_s, \underbrace{0, \ldots, 0}_{k-s \ \textrm{times}}  \big)
\]
We will now verify the properties of the construction. ~\\
\textbf{Proof of Items 1 and 2:} Both these items are straight forward from the construction.~\\
\textbf{Proof of Item 3: } By the second item of Lemma~\ref{lem:smoothing-1}, we have $\Pr_{x \sim \gamma_n} [\fsm^{'(i)}(x) \in \Delta_{k,k\delta/4}] \ge 1-\delta/2$. By applying Claim~\ref{clm:expectation-delta}, we obtain that whenever $\fsm^{'(i)}(x) \in \Delta_{k,k\delta/4}$, $f^{(i)}_1(x) \in \Delta_{k, O(k \delta + k \eta)}$. 
Thus, as long as $\eta \le \delta/k$, this proves Item 3 for $f^{(i)}_1$. The proof for $g^{(i)}_1$ is similar. ~\\
\textbf{Proof of Items  4 and 5:} We first observe that $\Pr_{x \sim \gamma_n}[\Vert \fsm^{'(i)}(x) - \fsm^{(i)}(x) \Vert_1  \le k \cdot \delta/4 ] \ge 1- \delta/2$. By applying Claim~\ref{clm:expectation-delta}, we obtain that $\Pr_{x \sim \gamma_n} [\Vert f^{(i)}_1(x) - \fsm^{(i)}(x) \Vert_1  \le O(k \delta + k \eta)] \ge 1- \delta/2$. However, note that by definition, 
$\Vert f^{(i)}_1(x) - \fsm^{(i)}(x) \Vert_\infty \le k$. This implies that $\mathbf{E}[\Vert \fsm^{'(i)} (x) - f^{(i)}_1(x) \Vert_1] = O(k\delta +  k \eta)$. As long as $\eta \le \delta/k$, we have $\mathbf{E}[\Vert \fsm^{(i)} (x) - f^{(i)}_1(x) \Vert_1] = O(k\delta)$. Combining with the guarantees of Corollary~\ref{corr:Boosting} yields Items 4 and 5. ~\\
\textbf{Proof of Item 6: } To prove Item 6, note that for any $1 \le s \le k$ and $\alpha_s \in [0,1]$, $$\arg \max \big(\underbrace{0, \ldots, 0}_{s-1 \ \textrm{times}} , p^{'(i)}_{f,s} - \alpha_s, \underbrace{0, \ldots, 0}_{k-s \ \textrm{times}}  \big) = \mathsf{PPF}_{p^{'(i)}_{f,s} - \alpha_s,s}.$$ 
Thus,  if we define $p^{(i)}_{s,j,1} = p^{'(i)}_{f,s}- \eta \cdot j$ and $p^{(i)}_{s,j,2} = p^{'(i)}_{g,s} - \eta \cdot j$, then 
\[
f_1^{(i)} = \sum_{s=1}^k \sum_{j=0}^{m} \frac{1}{m} \mathsf{PPF}_{p^{(i)}_{s,j,1} , s}  \  \textrm{and} \   g_1^{(i)} = \sum_{s=1}^k \sum_{j=0}^{m} \frac{1}{m} \mathsf{PPF}_{p^{(i)}_{s,j,2} , s},
\]
where $m = \lceil 1/\eta \rceil$. As $\eta \leq \delta /k$, $m = O(k/\delta)$. By Lemma~\ref{lem:smoothing-1}, $\mathsf{deg}(p^{'(i)}_{f,s})$ and $\mathsf{deg}(p^{'(i)}_{g,s})$ is at most $d' =  d \cdot \mathsf{poly}(k/\delta) \cdot \log^{d} (d k /\delta)$ where $d =2/t \cdot \log(dk/\delta)$ (coming from Corollary~\ref{corr:Boosting}). If we set $d_0(t,k,\delta) = d'$, then  $\mathsf{deg}(p^{'(i)}_{f,s})$ and $\mathsf{deg}(p^{'(i)}_{g,s})$ is at most $d_0(t,k,\delta)$. As $\mathsf{deg}(p^{(i)}_{s,j,1}) = \mathsf{deg}(p^{'(i)}_{f,s})$ and   $\mathsf{deg}(p^{(i)}_{s,j,2}) = \mathsf{deg}(p^{'(i)}_{g,s})$, this proves Item 6. (We can make the PPFs balanced by applying Fact~\ref{fact:balanced}). 
\end{proofof}

\section{Construction of junta polynomials} 
This section is dedicated to the proof of Lemma~\ref{lem:junta-construction}.  To prove this lemma, we will first recall the following important result from \cite{DMN16a} (Theorem~41 in that paper). 
\begin{theorem}~\label{thm:junta-construct}
Let $p_1, \ldots, p_{\ell}: \mathbb{R}^n \rightarrow \mathbb{R}$ be degree-$d$ polynomials and for $\delta>0$, the following two conditions: (i) For all $1 \le s \le \ell$, $\mathsf{Var}(p_s)=1$ and (ii) For all $1 \le s \le \ell $, $|\mathbf{E}[p_s]| \le \log^{d/2} (k \cdot d /\delta)$. For $1 \le s \le \ell$ and $t>0$, define  $u_s: \mathbb{R}^{2n} \rightarrow \mathbb{R}$
as follows: $u_s(x,y) = p_s(e^{-t} x + \sqrt{1-e^{-2t}} y)$. Then, there is an explicitly computable 
$n_0 = n_0(\ell, d, \xi)$ and polynomials $r_1, \ldots, r_{\ell}: \mathbb{R}^{n_0} \rightarrow \mathbb{R}$   with the following properties: For $1 \le s \leq \ell$, define $v_s: \mathbb{R}^{2n_0} \rightarrow \mathbb{R}$ as $v_s(x,y) = r_s(e^{-t} x + \sqrt{1-e^{-2t}} y)$. Then, for $ 1 \le s, s' \le \ell$, 
\begin{enumerate}
\item $\big|\Pr_{x \sim \gamma_n} [{p}_s \ge 0] - \Pr_{x \sim \gamma_n} [{r}_s \ge 0] \big|\le \xi$. 
\item $\big|\Pr_{x,y \sim \gamma_{n}} [{u}_s \ge 0] - \Pr_{x,y \sim \gamma_{n_0}} [{v}_s \ge 0] \big|\le \xi$. 
\item $\big|\Pr_{x \sim \gamma_n} [{p}_s \cdot {p}_{s'} \ge 0] - 
\Pr_{x \sim \gamma_{n_0}} [{r}_s \cdot {r}_{s'} \ge 0] \big|\le \xi$. 
\item $\big|\Pr_{x,y \sim \gamma_n} [{u}_s\cdot {u}_{s'} \ge 0] - 
\Pr_{x,y \sim \gamma_{n_0}} [{v}_s \cdot {v}_{s'} \ge 0] \big|\le \xi$. 
\item $\big|\Pr_{x,y \sim \gamma_n} [{p}_s\cdot {u}_{s'} \ge 0] - 
\Pr_{x,y \sim \gamma_{n_0}} [{v}_s \cdot {v}_{s'} \ge 0] \big|\le \xi$. 
\end{enumerate}
\end{theorem}
We now derive an additional property of 
the polynomials $\{p_s\}_{1 \le s \le \ell}$ and $\{r_s \}_{1 \le s \le \ell}$ defined in Theorem~\ref{thm:junta-construct} which will be useful later. 
\begin{corollary}~\label{corr:const-and-1}
Let $p_1, \ldots, p_{\ell}: \mathbb{R}^n \rightarrow \mathbb{R}$ and $u_1, \ldots, u_{\ell}: \mathbb{R}^n \rightarrow \mathbb{R}$ be as defined in Theorem~\ref{thm:junta-construct}. Then, for any $1 \le s, s' \le k$, 
\[
\big|\Pr_{x \sim \gamma_n} [({p}_s (x) \ge 0 ) \wedge ( {p}_{s'}(x) \ge 0)] - 
\Pr_{x \sim \gamma_{n_0}} [({r}_s (x) \ge 0 ) \wedge ( {r}_{s'}(x) \ge 0)] \big| \le 2\xi. 
\]
\end{corollary}
\begin{proof}
The main observation here is that if $A , B \not =0$, then 
$$
\mathbf{1}[A \ge 0 ] \cdot \mathbf{1}[ B \ge 0]  = \frac{1}{2} \big(\mathbf{1}[A \cdot B \ge 0] +\mathbf{1}[A  \ge 0]+\mathbf{1}[ B \ge 0]-1\big) . 
$$
Now, note that because  $p_s$, $p_{s'}$, $r_{s}$ and $r_{s'}$ are degree-$d$ polynomials, any of these functions  vanish over the Gaussian measure with probability $0$. Thus, 
\begin{eqnarray}
\Pr_{x \sim \gamma_n} [({p}_s (x) \ge 0 ) \wedge ( {p}_{s'}(x) \ge 0)]  &=& \frac{1}{2} \big(\Pr_{x \sim \gamma_n} [{p}_s (x) \ge 0 ]  +  \Pr_{x \sim \gamma_n} [  {p}_{s'}(x) \ge 0]   + \Pr_{x \sim \gamma_n} [  p_{s} \cdot {p}_{s'}(x) \ge 0]-1\big) \nonumber \\
\Pr_{x \sim \gamma_{n_0}} [({r}_s (x) \ge 0 ) \wedge ( {r}_{s'}(x) \ge 0)]  &=& \frac{1}{2} \big(\Pr_{x \sim \gamma_{n_0}} [{r}_s (x) \ge 0 ]  +  \Pr_{x \sim \gamma_{n_0}} [  {r}_{s'}(x) \ge 0]   + \Pr_{x \sim \gamma_{n_0}} [  r_{s} \cdot {r}_{s'}(x) \ge 0]-1\big) 
\nonumber \end{eqnarray}
Combining the above equations with items 1 and 3 in Theorem~\ref{thm:junta-construct} yields the corollary. 
\end{proof}
~\\
We now describe the proof of Lemma~\ref{lem:junta-construction}. 
\begin{proofof}{Lemma~\ref{lem:junta-construction}}
Let us consider the collection of degree-$d_0$ polynomials $\{p_{s,j,1}^{(i)}\}_{1 \le i \le \ell, 1 \le s \le k, 1 \le j \le m} \cup \{p_{s,j,2}^{(i)}\}_{1 \le i \le \ell,1 \le s \le k, 1 \le j \le m}$. We now apply Theorem~\ref{thm:junta-construct} to obtain polynomials $\{r_{s,j,1}^{(i)}\}_{1 \le i \le \ell,1 \le s \le k, 1 \le j \le m} \cup \{r_{s,j,2}^{(i)}\}_{1 \le i \le \ell,1 \le s \le k, 1 \le j \le m}$ with $\xi= \delta/(40 k^2)$. We now define 
\[
f^{(i)}_{\mathsf{junta}} = \sum_{s=1}^k \sum_{j=1}^m \frac{1}{m} \cdot \mathsf{PPF}_{r_{s,j,1}^{(i)},s}(x) \ , \  g^{(i)}_{\mathsf{junta}} = \sum_{s=1}^k \sum_{j=1}^m \frac{1}{m} \cdot \mathsf{PPF}_{r_{s,j,2}^{(i)},s}(x)
\]
We now verify the properties of the construction. ~\\
\textbf{Proof of Item 1: } Observe that for $1 \le s\le k$,  we have the following 
\[
\mathbf{E}[(f_{1}^{(i)}(x))_s] = \sum_{j=1}^m \frac{1}{m} \cdot \mathbf{E}_{x} [\mathsf{PPF}_{p_{s,j,1}^{(i)},s}(x)] = \sum_{j=1}^m \frac{1}{m} \cdot \Pr_{x} [p_{s,j,1}^{(i)}(x) \ge 0] 
\]
\[
\mathbf{E}[(f_{\mathsf{junta}}^{(i)}(x))_s] = \sum_{j=1}^m \frac{1}{m} \cdot \mathbf{E}_{x} [\mathsf{PPF}_{r_{s,j,1}^{(i)},s}(x)] = \sum_{j=1}^m \frac{1}{m} \cdot \Pr_{x} [r_{s,j,1}^{(i)}(x) \ge 0] 
\]
Thus, we obtain 
$$
\big|  \mathbf{E}[(f^{(i)}_{1}(x))_s] - \mathbf{E}[(f^{(i)}_{\mathsf{junta}}(x))_s] \big| \le \sup_{1 \le j \le m} \big| \Pr_{x} [p_{s,j,1}^{(i)}(x) \ge 0]-\Pr_{x} [r_{s,j,1}^{(i)}(x) \ge 0] \big| \le \xi . 
$$
The penultimate inequality follows by applying Theorem~\ref{thm:junta-construct} to 
$p^{(i)}_{s,j,1}$ and $r^{(i)}_{s,j,1}$. This immediately implies that $\Vert \mathbf{E}[f^{(i)}_{1}(x)] - \mathbf{E}[f^{(i)}_{\mathsf{junta}}(x)] \Vert_1 \le k \cdot \xi \le \delta$. 
The proof for 
$
\big|  \mathbf{E}[(g^{(i)}_{1}(x))_s] - \mathbf{E}[(g^{(i)}_{\mathsf{junta}}(x))_s] \big| \le \delta. 
$ is exactly identical. ~\\
\textbf{Proof of Item 2:} Like Item 1, we will only prove that $\Pr_{x} [f^{(i)}_{\mathsf{junta}}(x) \in \Delta_{k,\sqrt{\delta}}] \le \sqrt{\delta}$. The proof for $\Pr_{x} [g^{(i)}_{\mathsf{junta}}(x) \in \Delta_{k,\sqrt{\delta}}] \le \sqrt{\delta}$. To prove this, we first observe that for all $x$ both $f^{(i)}_1(x)$ and $f^{(i)}_{\mathsf{junta}}(x)$ always lie in the positive orthant and secondly, $\Vert f^{(i)}_1(x) \Vert_\infty, \Vert f^{(i)}_{\mathsf{junta}}(x) \Vert_\infty \le 1$. Next, 
\begin{eqnarray}
\mathbf{E}[(\Vert f^{(i)}_1(x) \Vert_1 - 1)^2 ]  &\le& \Pr_{x} [f^{(i)}_1(x) \in \Delta_{k,\delta}] \cdot  \delta^2 +  \Pr_{x} [f^{(i)}_1(x) \not \in  \Delta_{k,\delta}] \cdot k^2 \nonumber \\
&\le& \delta^2 + k^2 \cdot \delta. \label{eq:square-diff}
\end{eqnarray}
The first inequality uses $\sup_x \Vert f^{(i)}_1(x) \Vert_1 \le k$ and the second inequality uses $\Pr_{x} [f^{(i)}_1(x) \not \in  \Delta_{k,\delta}]  \le \delta$. Next, observe that 
\[
\Vert f^{(i)}_1(x) \Vert_1 = \sum_{s=1}^k \sum_{j=1}^m \frac{1}{m} \cdot \mathbf{1}[p^{(i)}_{s,j,1}(x) \ge 0] \ \  , \ \ \Vert f^{(i)}_{\mathsf{junta}}(x) \Vert_1 = \sum_{s=1}^k \sum_{j=1}^m \frac{1}{m} \cdot \mathbf{1}[r^{(1)}_{s,j,1}(x) \ge 0]
\]
This implies 
\begin{eqnarray}
( \Vert f^{(i)}_1(x) \Vert_1 - 1)^2  = \sum_{s=1}^k \sum_{s'=1}^k \sum_{j=1}^m \sum_{j'=1}^m \frac{1}{m^2} \mathbf{1}[p^{(i)}_{s,j,1}(x) \ge 0] \cdot \mathbf{1}[p^{(i)}_{s',j',1}(x) \ge 0] + 1 - \frac{2}{m}\sum_{s=1}^k  \sum_{j=1}^m   \mathbf{1}[p^{(i)}_{s,j,1}(x) \ge 0] .\label{eq:f-1}
\end{eqnarray}
\begin{eqnarray}
(\Vert f^{(i)}_{\mathsf{junta}}(x) \Vert_1 - 1)^2  = \sum_{s=1}^k \sum_{s'=1}^k \sum_{j=1}^m \sum_{j'=1}^m \frac{1}{m^2} \mathbf{1}[r^{(i)}_{s,j,1}(x) \ge 0] \cdot \mathbf{1}[r^{(i)}_{s',j',1}(x) \ge 0] + 1 - \frac{2}{m}\sum_{s=1}^k  \sum_{j=1}^m   \mathbf{1}[r^{(i)}_{s,j,1}(x) \ge 0] .\label{eq:f-junta-1}
\end{eqnarray}
Recall that by construction, we have
\begin{equation}
\sup_{1 \le s \le k, \ 1 \le j \le m} \big| \Pr_{x} [p_{s,j,1}^{(i)}(x) \ge 0]-\Pr_{x} [r_{s,j,1}^{(i)}(x) \ge 0] \big| \le \xi \label{eq:diff-p-r}
\end{equation}
Applying Corollary~\ref{corr:const-and-1}, we also obtain 
\begin{equation}
\sup_{1 \le s,s' \le k, \ 1 \le j,j' \le m} \big| \Pr_{x} [(p_{s,j,1}^{(i)}(x) \ge 0) \wedge (p_{s',j',1}^{(i)}(x) \ge 0)]-\Pr_{x} [(r_{s,j,1}^{(i)}(x) \ge 0) \wedge (r_{s',j',1}^{(i)}(x) \ge 0)] \big| \le 2\xi. \label{eq:diff-int-p-r}
\end{equation}
Applying (\ref{eq:diff-p-r}) and (\ref{eq:diff-int-p-r}) to (\ref{eq:f-1}) and (\ref{eq:f-junta-1}), we obtain 
\[
\big| \mathbf{E}[(\Vert f_{\mathsf{junta}}^{(i)}(x) \Vert_1 - 1)^2 ]- \mathbf{E}[(\Vert f^{(i)}_1(x) \Vert_1 - 1)^2 ]\big| \le 2  k^2  \cdot \xi + 2  k \cdot \xi \le \delta. 
\]
Combining this with (\ref{eq:square-diff}), we obtain 
$
\mathbf{E}[(\Vert f^{(i)}_{\mathsf{junta}}(x) \Vert_1 - 1)^2 ] \le 2k^2 \cdot \delta.
$ 
Applying Markov's inequality, we obtain that 
$\Pr[| \ \Vert f^{(i)}_{\mathsf{junta}}(x) \Vert_1 - 1| > k\sqrt{\delta}] \le 2 k \sqrt{\delta}$. Since $f^{(i)}_{\mathsf{junta}}(x)$ lies in the positive orthant for any $x$, this proves Item 2.~\\
\textbf{Proof of Item 3: }  To prove Item 3, we observe that for any $1 \le s_1, s_2 \le k$, 
\begin{eqnarray}
\mathbf{E}[f_{1,s_1} P_t g_{1,s_2}] &=&  \frac{1}{m^2}\sum_{j=1}^m \sum_{j'=1}^m \mathbf{E}\big[\mathsf{PPF}_{p^{(1)}_{s_1,j_1}}(x) P_t \ \mathsf{PPF}_{p^{(2)}_{s_2,j_2}}(x)\big] \nonumber \\
&=&  \frac{1}{m^2}\sum_{j=1}^m \sum_{j'=1}^m \mathbf{E}_{x,y}\big[\mathsf{PPF}_{p^{(1)}_{s_1,j_1}}(x) \mathsf{PPF}_{p^{(2)}_{s_2,j_2}} (e^{-t} x + \sqrt{1-e^{-2t}} y)\big] \nonumber\\
&=&\frac{1}{m^2}\sum_{j=1}^m \sum_{j'=1}^m \Pr_{x, y} [
(p^{(1)}_{s_1,j_1}(x) \ge 0) \wedge (p^{(2)}_{s_2,j_2}(e^{-t} x + \sqrt{1-e^{-2t}} y) \ge 0)] \nonumber \\
&=& \frac{1}{m^2}\sum_{j=1}^m \sum_{j'=1}^m \Pr_{x, y} [
(p^{(1)}_{s_1,j_1}(x) \ge 0) \wedge (u^{(2)}_{s_2,j_2}(e^{-t} x + \sqrt{1-e^{-2t}} y) \ge 0)] \label{eq:fpg-1}.
\end{eqnarray}
Likewise, we can obtain 
\begin{equation}\label{eq:fpg-2}
\mathbf{E}[f_{\mathsf{junta},s_1} P_t g_{\mathsf{junta},s_2}] = \frac{1}{m^2} \sum_{j=1}^m \sum_{j'=1}^m \Pr_{x, y} [
(r^{(1)}_{s_1,j_1}(x) \ge 0) \wedge (v^{(2)}_{s_2,j_2}(e^{-t} x + \sqrt{1-e^{-2t}} y) \ge 0)].
\end{equation}
Combining (\ref{eq:fpg-1}) and (\ref{eq:fpg-2})  with Item 5 in Theorem~\ref{thm:junta-construct} yields
\[
\big| \mathbf{E}[f_{1,s_1} P_t g_{1,s_2}] -\mathbf{E}[f_{\mathsf{junta},s_1} P_t g_{\mathsf{junta},s_2}] \big | \le \xi. 
\]
This finishes the proof. 
\end{proofof}

\subsection*{Acknowledgments}
We thank
Pritish Kamath, Badih Ghazi and Madhu Sudan for pointing out that the $\ell=1$ case
of Theorem~\ref{thm:junta-strong} is not sufficient to derive Theorem~\ref{thm:junta}. 
(An earlier version of this paper incorrectly claimed that it was.)
We also thank the anonymous reviewers
who pointed out the same gap.

\bibliography{allrefs}
\bibliographystyle{alpha}
\appendix

\section{Reduction from arbitrary $\bP$ to the Gaussian case}~\label{section:GKS} 
We first restate Theorem~\ref{thm:junta} below. 
\begin{theorem*}{\textbf{2}}
Suppose there exist
  $f, g: \mathcal{Z}^n \to [k]$ such that $(f(\bX^n), g(\bY^n)) \sim \bQ$.
  Then, there exist $n_0 = n_0 (|\mathbf{P}|, \delta)$ and $f_{\mathsf{junta}},
  g_{\mathsf{junta}} : \mathcal{Z}^{n_0} \rightarrow [k]$ such that $\bQ$ and the
  distribution of $(f_{\mathsf{junta}}(\bX^{n_0}), g_{\mathsf{junta}}(\bY^{n_0}))$ are
  $\delta$-close in total variation distance. Moreover, $n_0$ is computable. Further, the functions $f_{\mathsf{junta}}$ and $g_{\mathsf{junta}}$ can be explicitly computed. 
\end{theorem*}
Next, we restate Theorem~\ref{thm:junta-strong}. 
\begin{theorem*}{\textbf{5}}
Let $\bP = (\bX,\bY) = \mathbf{G}_{\rho,2}$ and let $f^{(1)}, \ldots, f^{(\ell)}: \mathbb{R}^n \rightarrow [k]$ and $g^{(1)}, \ldots, g^{(\ell)}: \mathbb{R}^n \rightarrow [k]$ where we  define $\bQ_{i,j}$ as 
$\bQ_{i,j} = (f^{(i)}(\bX^n), g^{(j)}(\bY^n))$. 
Then, for every $\delta>0$, there is an explicitly defined constant 
$n_0 = n_0(\ell, k, \delta)$ and explicitly defined functions $f^{(1)}_{\mathsf{junta}}, \ldots, f^{(\ell)}_{\mathsf{junta}}: \mathbb{R}^{n_0} \rightarrow [k]$  and $g^{(1)}_{\mathsf{junta}}, \ldots, g^{(\ell)}_{\mathsf{junta}}: \mathbb{R}^{n_0} \rightarrow [k]$ such that for every $1 \le i, j \le \ell$, 
$\dtv((f^{(i)}_{\mathsf{junta}}(\bX^{n_0}), g^{(j)}_{\mathsf{junta}}(\bY^{n_0})), \bQ_{i,j}) \le \delta$. 
\end{theorem*}
The main purpose of this section is to show how proving Theorem~\ref{thm:junta} reduces to proving Theorem~\ref{thm:junta-strong}. While the reduction essentially follows just going over the steps in \cite{GKS16} \emph{mutatis mutandis} (which in turn relies on standard tools from Boolean function analysis), for the purposes of clarity, we give a brief overview of the reduction here. 

First, let us fix some notation. \begin{enumerate} 
\item We recall the notion of maximal correlation coefficient: Namely, given a probability space $(\bX, \bY)$, we let $\rho(\bX, \bY)$ be defined as 
$$
\rho(\bX, \bY) = \sup \mathbf{E}[\Psi_1(\bX) \cdot \Psi_2 (\bY)],
$$
where the supremum is taken over all functions which satisfy $\mathbf{E}[\Psi_1(\bX)]= \mathbf{E}[\Psi_2(\bY)]=0$ and 
$\mathsf{Var}[\Psi_1(\bX)] = \mathsf{Var}[\Psi_2(\bY)]=1$. 
\item For a given set $H \subseteq [n]$, $x_H \in \bX^{|H|}$ and function $f: \bX^n \rightarrow \mathbb{R}^k$, we let $f(x_H, .) : \bX^{[n] \setminus H} \rightarrow \mathbb{R}^k$ denote the function obtained by fixing the coordinates of $f$ in $H$ to $x_H$.  
\end{enumerate}

As we have stated before, for the case $k=2$, Ghazi, Kamath and Sudan~\cite{GKS16} reduce Theorem~\ref{thm:junta} for the general $\bP$ case to the case when $\bP = \mathbf{G}_{\rho,2}$. In other words, for $k=2$,  \cite{GKS16} reduces Theorem~\ref{thm:junta} for the general $\bP$ case to Theorem~\ref{thm:junta-strong} with $\ell=1$. We now give a sketch of why Theorem~\ref{thm:junta} reduces to Theorem~\ref{thm:junta-strong} for $k>1$.

\textbf{Overview of the reduction: }Using the regularity lemma for low-degree polynomials~\cite{DSTW:10, DDS14} and other ideas from Boolean function analysis (along the lines of  \cite{GKS16}), one can easily show the following: Let $\tau>0$ be any error parameter. Then, there exists a set $H \subseteq [n]$ such that $|H| = O_{\tau,|\bP|,k}(1)$ and for $(x_H, y_H) \sim (\bX, \bY)^H$, with probability $1-\tau$, the following holds: The functions $f(x_H, \cdot)$ and $g(y_H, \cdot)$ are \emph{low-influence} functions namely, 
$$
\max_{i \in [n] \setminus H} \mathsf{Inf}_i(f(x_H, \cdot)) \leq \tau, \ \ \max_{i \in [n] \setminus H} \mathsf{Inf}_i(g(y_H, \cdot)) \leq \tau. 
$$
In the above definition, for $f: \mathbb{R}^n \rightarrow \mathbb{R}^k$, we let $\mathsf{Inf}_i(f)$
denotes the quantity
\[
\mathsf{Inf}_i(f) = \sum_{i \in S: S \in \mathbb{Z}^{\ast n}} \Vert \widehat{f}(S) \Vert_2^2,
\]
where $\widehat{f}(S)$ denotes the Hermite coefficient of $f$ corresponding to $S$. Note that this is the standard definition of ``influence" from Boolean function analysis (see~\cite{ODonnell:book, Mossel2010}). In fact, one can also additionally assume that every coordinate of $f$ and $g$ is essentially a low-degree polynomial. 

To understand why the low-influence condition is useful, let $\bP_{G} = \mathbf{G}_{\rho,2}$ where $\rho = \rho(\bX,\bY)$.
Further, let $(\bX_G, \bY_G) = \bP_G$. 
Likewise, let $\tilde{f}(x_H, \cdot)$ (resp. $\tilde{g}(y_H, \cdot)$) be the multilinear extension of $f(x_H, \cdot)$ (resp. ${g}(y_H, \cdot)$) to the Gaussian space.
Then, the invariance principle of Mossel \emph{et~al.}~\cite{MOO10, Mossel2010} shows that as long as $\tau$ is chosen to be sufficiently small in $\delta$, for any pair $(x_H, y_H)$ where $f(x_H, \cdot)$ and $g(y_H, \cdot)$ are low-influence functions, the following holds: 
 \[ 
 \dtv(
 (\tilde{f}(x_H, \bX_G^{[n] \setminus H}), \tilde{g}(y_H, \bY_G^{[n] \setminus H})) , (f(x_H, \bX^{[n] \setminus H}), f(y_H, \bY^{[n] \setminus H})) \le \delta/4.
 \]
  Note that the total number of $(x_H, y_H)$ pairs is bounded by $|\mathsf{supp}(\bP)|^{2|H|}$. Let us denote this number by $\mathbf{N}_{sup}$.  
 By applying Theorem~\ref{thm:junta-strong},we obtain that for any $\delta>0$, there is $n_0 = n_0(\mathbf{N}_{sup},k,\delta)$ such that 
 corresponding to every function $\tilde{f}(x_H, \cdot)$ (resp. $\tilde{g}(y_H, \cdot)$ ), 
 there is 
 a function 
 $\underline{f_{x_H}} : \mathbb{R}^{n_0} \rightarrow [k]$ (resp. $\underline{g_{y_H}}: \mathbb{R}^{n_0} \rightarrow [k]$ ) such that
 \[
 \dtv\big(\big(\tilde{f}(x_H, \bX_G^{[n] \setminus H}), \tilde{g}(y_H, \bY_G^{[n] \setminus H}\big), \big(\underline{f}_{x_H}( \bX_G^{n_0}), \underline{g}_{y_H}( \bY_G^{n_0})\big) \big) \le \delta/4. 
 \]
Note that here we are crucially using the fact that Theorem~\ref{thm:junta-strong} is valid for an arbitrary $\ell \ge 1$ and not just $\ell=1$. 
Let us define $m_0 = n_0 \cdot (1/\kappa^2)$. We next define $\underline{f}_{\mathsf{low}, \ x_H} : \mathbb{R}^{m_0} \rightarrow [k]$ as 
\[
\underline{f}_{\mathsf{low}, \ x_H}  \big(x_{1,1}, \ldots, x_{n_0, \kappa^{-2}} \big) = \underline{f}_{x_H} \big( \kappa \cdot (x_{1,1} + \ldots + x_{1,\kappa^{-2}}) , \ldots, \kappa \cdot (x_{n_0,1} + \ldots + x_{n_0,\kappa^{-2}})\big). 
\]
\[
\underline{g}_{\mathsf{low}, \ y_H}  \big(y_{1,1}, \ldots, y_{n_0, \kappa^{-2}} \big) = \underline{g}_{y_H} \big( \kappa \cdot (y_{1,1} + \ldots + y_{1,\kappa^{-2}}) , \ldots, \kappa \cdot (y_{n_0,1} + \ldots + y_{n_0,\kappa^{-2}})\big). 
\]
 From the definition of $\underline{f}_{x_H}$ and $\underline{g}_{y_H}$, it easily follows that, 
 \[
 \big(\underline{f}_{x_H}( \bX_G^{n_0}), \underline{g}_{y_H}( \bY_G^{n_0})\big) = \big(\underline{f}_{\mathsf{low}, x_H}( \bX_G^{m_0}), \underline{g}_{\mathsf{low}, y_H}( \bY_G^{m_0})\big)
 \]
  Let
 ${f}_{\mathsf{low}, x_H}$ and ${g}_{\mathsf{low}, y_H}$ denote the multilinear extensions of $\underline{f}_{\mathsf{low}, x_H}$ and $\underline{g}_{\mathsf{low}, y_H}$ to the space $(\bX^{m_0}, \bY^{m_0})$. Observe that the functions $\underline{f}_{\mathsf{low}, x_H}$ and $\underline{g}_{\mathsf{low}, y_H}$ have influence bounded by $\kappa$. Thus, as long as $\kappa$ is chosen to be a sufficiently small function of $\delta$, the invariance principle~\cite{Mossel2010} implies that 
 \[
 \dtv \big( ({f}_{\mathsf{low}, x_H} ( \bX^{m_0}),{g}_{\mathsf{low}, y_H} (\bY^{m_0})), (\underline{f}_{\mathsf{low}, x_H}( \bX_G^{m_0}), \underline{g}_{\mathsf{low}, y_H}( \bY_G^{m_0}))\big)  \le \delta/4. 
 \]
 Combining the above three equations, we get that 
 \[
 \dtv \big( ({f}_{\mathsf{low}, x_H} ( \bX^{m_0}),{g}_{\mathsf{low}, y_H} (\bY^{m_0})), ({f}(x_H, \bX^{[n] \setminus H}), {g}(y_H, \bY^{[n] \setminus H}))\big)  \le \frac{3 \delta}{4}.
 \]
 

With this, we define functions $f_{\mathsf{junta}}: \mathbb{R}^{m_0 + |H|} \rightarrow [k]$ and $g_{\mathsf{junta}}: \mathbb{R}^{m_0 + |H|} \rightarrow [k]$ as follows. Split $x \in \mathbb{R}^{m_0 +H}$ as $(x_H, x_{m_0})$ and $y \in \mathbb{R}^{m_0 +H}$ as $(y_H, y_{m_0})$.
\[
f_{\mathsf{junta}}(x_H, x_{m_0}) = {f}_{\mathsf{low}, x_H} (x_{m_0}) ; \ g_{\mathsf{junta}}(y_H, y_{m_0}) = {g}_{\mathsf{low}, y_H} (y_{m_0}). 
\] 
This immediately implies $$\dtv((f(\bX^n), g(\bY^n)),(f_{\mathsf{junta}}(\bX^{m_0+|H|}), g_{\mathsf{junta}}(\bY^{m_0 +|H|})) \le \frac{3\delta}{4} + \tau. $$ 
Once we choose $\tau \le \delta/4$, the reduction is complete. 
\end{document}